\documentclass[11pt,a4paper]{article}
\usepackage[utf8]{inputenc}
\usepackage[T1]{fontenc}
\usepackage[english]{babel}
\usepackage{tikz}\usetikzlibrary{decorations.pathreplacing, arrows.meta}
\usepackage{geometry}
\usepackage{authblk}
\usepackage{amsmath,amssymb,amsthm}

\newtheorem{proposition}{Proposition}[section]
\newtheorem{definition}[proposition]{Definition}
\newtheorem{example}[proposition]{Example}
\newtheorem{remark}[proposition]{Remark}

\usepackage{mathtools}
\usepackage{stmaryrd}
\usepackage{float}

\numberwithin{equation}{section}

\usepackage{graphicx}
\usepackage{microtype}
\usepackage{xcolor}
\usepackage[colorlinks=true,
pdfstartview=FitV,
linkcolor= blue,
citecolor= blue,
urlcolor= blue,
hyperindex=true,
hyperfigures=false]
{hyperref}

\usepackage{cite}

\newcommand{\rc}{\check{r}}

\title{\textbf{Integrable multi-species SSEP\\ with reactive particle species}}
\author[1]{Mathieu Dabrowski\thanks{email: \href{mailto:mathieudabrowski@outlook.com}{mathieudabrowski@outlook.com}}\thanks{Corresponding author.}}
\author[2]{Loïc Poulain d’Andecy\thanks{email: \href{mailto:loic.poulain-dandecy@univ-reims.fr}{loic.poulain-dandecy@univ-reims.fr}}}
\author[1]{Eric Ragoucy\thanks{email: \href{mailto:ragoucy@lapth.cnrs.fr}{ragoucy@lapth.cnrs.fr}}}
\affil[1]{Laboratoire d'Annecy de Physique Théorique, 9 Chemin de Bellevue - BP 110 - Annecy-le-Vieux - F-74941 Annecy Cedex - France}
\affil[2]{Laboratoire de Mathématiques de Reims, UMR CNRS 9008, Universit\'e Reims Champagne-Ardenne, 51687 Reims - France}

\date{}
\begin{document}

\maketitle

\begin{abstract}
	We investigate integrable exclusion Markov processes constructed from set-theoretical solutions of the Yang-Baxter Equation (YBE) that generalise the multi-species Symmetric Simple Exclusion Process (SSEP). We first introduce the process called the $ (p,1) $-SSEP that is analogous to the multi-species SSEP but with an extra particle species qualified as reactive. Reactive species are able to evaporate and condense by pairs or to transform by pairs depending on the interpretation.
	We provide a full combinatorial study of the sectors in the periodic case. We prove the integrability of the process in the periodic case and in the open case for two types of boundaries that we introduce using Baxterisations of solutions of the reflection equation. Next we move on to the process called $ (p,q) $-SSEP, i.e. the multi-species SSEP with an arbitrary number of reactive particle species. We prove its integrability in the periodic case and, for the open case, we introduce integrable boundaries generalising the ones considered for the $ (p,1) $-SSEP.
	Finally, we define physical quantities relevant to the models and we compute them in the non-equilibrium stationary state of the open $ (p,q) $-SSEP for one type of integrable boundaries.

\end{abstract}


\section{Introduction}

Exclusion Markov processes historically appeared in various areas of science, such as biology \cite{mac1968} to describe protein synthesis, social sciences \cite{Schelling1971} to model segregation behaviours, economics \cite{Follmer1974} to take into account random behaviours of economic agents, physics \cite{Kawa1966} to explain the diffusion of spins in ferromagnets... In mathematics, it has grown into a significant branch of probability theory \cite{Liggett1999,Liggett2005}. 

The Symmetric Simple Exclusion Process (SSEP) \cite{Spitzer1970} describes particles on a lattice that can jump with equal probabilities one site to their left or to their right with the constraint that no two particles can occupy the same site. This model is well understood on a lattice with a countable number of sites \cite{Ferrari1991,Ferrari1991_2,Liggett1999} or in the hydrodynamic limit \cite{Kipnis1999}. Several directions have been explored to generalise this model such as adding a weak asymmetry between left or right jump probabilities \cite{Kipnis1989,DeMasi89}, or adding boundaries that drive the system out-of-equilibrium \cite{Bertini2003,Derrida2020}, or allowing more than one particle per site \cite{Kipnis1994, GF2021}.
Another way of generalising the SSEP is to add several types of particles. Historically, the addition of an extra particle species appeared for totally biased jump probabilities \cite{Derrida1993, Ferrari1994}. The unbiased multi-species SSEP was later studied in \cite{Simpson2009, Nagahata2011}. As in the SSEP, the multi-species process can be driven out-of-equilibrium with the addition of boundaries that can inject or extract particles \cite{vanicat2017_2}.

The classical stochastic dynamics of exclusion Markov processes is closely related to the dynamics of quantum spin chains \cite{Kadanoff1968, Alcaraz1994, Schutz2001}. More precisely, transition probabilities of Markov processes can be stored in a Markov matrix. For certain processes, this Markov matrix can be related to the quantum Hamiltonian of a spin chain. In particular, this is the case of the Markov matrix of the SSEP which is conjugated to the Hamiltonian of the Heisenberg ferromagnet \cite{Alexander1978}. These relationships allow the application of the tools developed in the context of quantum spin systems to exclusion Markov processes such as the Bethe ansatz and related algebraic methods \cite{FST, Gwa1992, faddeev1996, Sasamoto1997, Golinelli2004, Gier2005}.

Some spin systems have the property of being integrable \cite{Baxter1982}, that is to say, they have a set of conserved charges that makes those systems exactly solvable. Those conserved charges can be obtained from commuting transfer matrices that admit different forms depending on the boundary conditions of the system \cite{Sklyanin1988}. An exclusion Markov process can also be integrable \cite{Popkov2002}. We can take advantage of this property to get information about its stationary and dynamical behaviours \cite{Crampe2014}.

One approach to construct an integrable exclusion Markov process is to build a transfer matrix that will generate the Markov matrix of the process. This can be done using solutions of the (quantum) Yang-Baxter Equation (YBE) \cite{Yang1967,Baxter1982} which are called $ R $ matrices. They correspond to automorphisms of $ V \otimes V $ where $ V $ is a vector space and their classification is still an open problem. As Drinfeld suggested  in \cite{Drinfeld1992} a strategy different from the quantum group approach \cite{Drinfeld1988,Jimbo1985} is to focus on $ R $ matrices which are induced by a linear extension of a bijection $ \mathcal{R} $ of $ X \times X $ where $ X $ is a set. We say in that case that $ \mathcal{R} $ is a set-theoretical solution of the YBE. In particular, the transfer matrix of the periodic multi-species SSEP comes from the set-theoretical map $ P $ acting on $ X \times X $ as $ P(i,j)=(j,i) $.
A particular class of set-theoretical solutions which has been extensively studied is the one of non-degenerate involutive solutions \cite{etingof1998,Gateva1998,cedo2008,Smoktunowicz2020,Doikou_2021_2} and algebraic structures such as braces have been created to construct such solutions \cite{Rump2005,Rump2007,Cedo2012,cedo2018}. 

In a preceding paper \cite{ragoucy2026}, we constructed and studied Markov models from a class of set-theoretical solutions of the YBE called involutive Lyubashenko solutions. These solutions are indexed by an arbitrary permutation of $X$ and are the simplest solutions. It was shown \cite{ragoucy2026} that the periodic Markov models constructed from these solutions can be identified with a twisted periodic multispecies SSEP, the twist being given in terms of the permutation of $X$. Therefore, to go beyond SSEP and twisted SSEP, we need to consider set-theoretical solutions outside the class of Lyubashenko solutions and this is what we do in the present work.

\medskip

In this paper, we define new integrable Markov processes on one-dimensional lattices (with finite number of sites) with periodic or open boundary conditions. They generalise the multi-species SSEP \cite{vanicat2017_2} by the addition of a new type of particle species that we call {\it{reactive particles}}. This name reflects the fact that these particle species can transform into one another during the process. These processes are constructed from a class of set-theoretical solutions of the YBE which are not of Lyubashenko type.

We denote by $ (p,q) $-SSEP these processes, where $ p $ is the number of SSEP-like particle species referred to as singlets and $ q $ is the number of pairs of reactive particle species. Processes sharing some similarities with the $ (p,q) $-SSEP have been studied in the literature. For instance, the DiSSEP \cite{vanicat2016} has particle species that can evaporate and condense by pairs and are similar to (but different from) the reactive species of the $ (p,q) $-SSEP. 
Other examples of similar but different process are the reaction-diffusion process with several species introduced in \cite{Alcaraz1993} where particles can diffuse, interchange and annihilate, and the exclusion diffusion-reaction system with two reactive particle species introduced in \cite{Schutz1994} which displays an $ SU(2) $ symmetry on a subspace containing its stationary states.
The distinction between these models will be discussed in the body of the text.

The $(p,q)$-SSEP models come with the following possible physical interpretations (possibly among others).

In a first interpretation, the particles in a reactive pair, denoted $\{a^-,a^+\}$, can transform one into the other during the process, through the reaction $a^-a^-\to a^+a^+$ or its inverse. However, two adjacent $a^-$ and $a^+$ cannot exchange their position. This is illustrated in Figure \ref{fig:annihi_process_(p,q)}. The relations between singlet particles, or between one singlet particle and one reactive particle, or between two reactive particles of two different pairs, are the usual ones of the SSEP (they can only exchange their position). Note that this interpretation is valid for any number of pairs of reactive species. 
\begin{figure}[H]
    \centering
    % \resizebox{\columnwidth}{!}{
    \begin{tikzpicture}[
        scale=0.8, transform shape,
        % Singlets: Black fill, white text, 14pt
        singlet/.style={circle, draw=black, fill=black, text=white, minimum size=14pt, inner sep=0pt, font=\small\bfseries},
        % Reactive: Gray fill, thick border, black text, 14pt, font matches singlet
        reactive/.style={circle, draw=black, thick, fill=gray!20, text=black, minimum size=14pt, inner sep=0pt, font=\small\bfseries},
        % Superscript labels placed outside the contour at a 45-degree angle
        plus/.style={label={[font=\scriptsize\bfseries, inner sep=1pt, label distance=-2pt]45:$+$}},
        minus/.style={label={[font=\scriptsize\bfseries, inner sep=1pt, label distance=-2pt]45:$-$}},
        pm/.style={label={[font=\scriptsize\bfseries, inner sep=1pt, label distance=-2pt]45:$\pm$}},
        mp/.style={label={[font=\scriptsize\bfseries, inner sep=1pt, label distance=-2pt]45:$\mp$}}
    ]

    % =========================================================
    % MAIN LATTICE
    % =========================================================
    % Periodic extensions and main lattice line
    \node at (-0.8, 2.6) {$\cdots$};
    \draw[dashed, thick] (-0.5, 2.5) -- (0, 2.5);
    \draw (0,2.5) -- (14,2.5);
    \draw[dashed, thick] (14, 2.5) -- (14.5, 2.5);
    \node at (14.8, 2.6) {$\cdots$};

    % Lattice ticks
    \foreach \x in {0,...,14} \draw (\x,2.5) -- (\x,2.7);

    % --- PARTICLES ---
    % 1. Singlet-Singlet
    \node[singlet]         at (0.5, 2.8) {i};
    \node[singlet]         at (1.5, 2.8) {j};
    
    % 2. Singlet-Reactive
    \node[singlet]         at (2.5, 2.8) {k};
    \node[reactive, pm]    at (3.5, 2.8) {a};
    
    % 3. Reactive-Reactive (Same sign)
    \node[reactive, pm]    at (4.5, 2.8) {a};
    \node[reactive, pm]    at (5.5, 2.8) {b};
    
    % 4. Reactive-Reactive (Opposite sign)
    \node[reactive, pm]    at (6.5, 2.8) {a};
    \node[reactive, mp]    at (7.5, 2.8) {b};
    
    % 5. Blocked Exchange
    \node[reactive, pm] at (8.5, 2.8) {a};
    \node[reactive, mp]  at (9.5, 2.8) {a};

    % 6. Creation-Annihilation 1
    \node[reactive, minus] at (10.5, 2.8) {a};
    \node[reactive, minus] at (11.5, 2.8) {a};
    
    % 7. Creation-Annihilation 2
    \node[reactive, plus]  at (12.5, 2.8) {b};
    \node[reactive, plus]  at (13.5, 2.8) {b};

    % =========================================================
    % ACTIONS: Possible Local Changes
    % =========================================================
    % --- Valid Exchanges (Rate 1) ---
    % Singlet-Singlet
    \draw[<->, >=stealth, thick] (0.6, 3.2) to[bend left=45] node[above] {1} (1.4, 3.2);
    % Singlet-Reactive
    \draw[<->, >=stealth, thick] (2.6, 3.2) to[bend left=45] node[above] {1} (3.4, 3.2);
    % Reactive-Reactive (a^\pm <-> b^\pm)
    \draw[<->, >=stealth, thick] (4.6, 3.2) to[bend left=45] node[above] {1} (5.4, 3.2);
    % Reactive-Reactive (a^\pm <-> b^\mp)
    \draw[<->, >=stealth, thick] (6.6, 3.2) to[bend left=45] node[above] {1} (7.4, 3.2);

    % --- Blocked Exchanges (Crossed out) ---
    % a^- a^+ block (Black cross)
    \draw[<->, >=stealth, thick] (8.6, 3.2) to[bend left=45] coordinate[midway] (m1) (9.4, 3.2);
    \draw[thick] ([xshift=-4pt,yshift=-4pt]m1) -- ([xshift=4pt,yshift=4pt]m1);
    \draw[thick] ([xshift=-4pt,yshift=4pt]m1) -- ([xshift=4pt,yshift=-4pt]m1);

    % --- Creation-Annihilation Transformations ---
    % Upper braces grouping original particles
    \draw[decorate, decoration={brace, amplitude=4pt}, thick] (10.2, 3.2) -- (11.8, 3.2);
    \draw[decorate, decoration={brace, amplitude=4pt}, thick] (12.2, 3.2) -- (13.8, 3.2);

    % Vertical short double arrows (shifted up to create vertical space above braces)
    \draw[->, >=stealth, thick] (11.0, 3.5) -- (11.0, 3.9) node[midway, right] {1};
    \draw[->, >=stealth, thick] (13.0, 3.5) -- (13.0, 3.9) node[midway, right] {1};

    % Upper braces opening upward (mirror points it downward to cup the upper particles)
    \draw[decorate, decoration={brace, amplitude=4pt, mirror}, thick] (10.2, 4.2) -- (11.8, 4.2);
    \draw[decorate, decoration={brace, amplitude=4pt, mirror}, thick] (12.2, 4.2) -- (13.8, 4.2);

    % Transformed particles on top
    \node[reactive, plus]  at (10.5, 4.6) {a};
    \node[reactive, plus]  at (11.5, 4.6) {a};
    \node[reactive, minus] at (12.5, 4.6) {b};
    \node[reactive, minus] at (13.5, 4.6) {b};
    \end{tikzpicture}
    \caption{Illustration of the dynamical rules in the transformation interpretation of the $ (p,q) $-SSEP. Particle species in black labeled by $i,j,k\in \llbracket 1,p \rrbracket $ correspond to singlets and particle species in grey labeled $ a^{\pm} $ and $ b^{\pm} $ with $ a,b\in \llbracket 1,q \rrbracket $ and $ a\neq b $ to the reactive ones.}
    \label{fig:annihi_process_(p,q)}
\end{figure}

In the special case of the $(p,1)$-SSEP, we have another interpretation. Denoting  by $\{0,\bar{0}\}$ the reactive pair, the species $0$ is interpreted as an empty site and the species $\bar{0}$ can be interpreted as evaporating/condensing impurities as illustrated in Figure \ref{fig:evap-cond}.
\begin{figure}[H] 
    \centering 
     \begin{tikzpicture}[ 
        scale=0.8, transform shape, 
        % Species 1: Black fill, white text (moves freely, exchanges) 
        p1/.style={circle, draw=black, fill=black, text=white, minimum size=14pt, inner sep=0pt, font=\small\bfseries}, 
        % Species 2: Gray fill, thick border, black text (exchanges, condenses/evaporates) 
        p2/.style={circle, draw=black, thick, fill=gray!20, text=black, minimum size=14pt, inner sep=0pt, font=\small\bfseries} 
    ] 

    % ========================================================= 
    % MAIN LATTICE
    % ========================================================= 

    % Periodic extensions and main lattice line (17 sites) 
    \node at (-0.8, 2.6) {$\cdots$}; 
    \draw[dashed, thick] (-0.5, 2.5) -- (0, 2.5); 
    \draw (0,2.5) -- (17,2.5); 
    \draw[dashed, thick] (17, 2.5) -- (17.5, 2.5); 
    \node at (17.8, 2.6) {$\cdots$}; 

    % Lattice ticks 
    \foreach \x in {0,...,17} \draw (\x,2.5) -- (\x,2.7); 

    % --- 1. Allowed Movement: Species 1 (0i <-> i0) --- 
    % Particle at 1.5, with empty sites at 0.5 and 2.5
    \node[p1] at (1.5, 2.8) {i}; 
    \draw[<-, >=stealth, thick] (0.6, 3.2) to[bend left=50] node[above] {1} (1.4, 3.2); 
    \draw[->, >=stealth, thick] (1.6, 3.2) to[bend left=50] node[above] {1} (2.4, 3.2); 

    % --- 2. Blocked Movement: Solitary Species 2 --- 
    % Empty site at 3.5 (next to 2.5), Particle at 4.5, empty site at 5.5
    \node[p2] at (4.5, 2.8) {$\bar{0}$}; 
    
    % Blocked left movement
    \draw[<-, >=stealth, thick] (3.6, 3.2) to[bend left=50] coordinate[midway] (m1) (4.4, 3.2); 
    \draw[thick] ([xshift=-4pt,yshift=-4pt]m1) -- ([xshift=4pt,yshift=4pt]m1);
    \draw[thick] ([xshift=-4pt,yshift=4pt]m1) -- ([xshift=4pt,yshift=-4pt]m1);

    % Blocked right movement
    \draw[->, >=stealth, thick] (4.6, 3.2) to[bend left=50] coordinate[midway] (m2) (5.4, 3.2); 
    \draw[thick] ([xshift=-4pt,yshift=-4pt]m2) -- ([xshift=4pt,yshift=4pt]m2);
    \draw[thick] ([xshift=-4pt,yshift=4pt]m2) -- ([xshift=4pt,yshift=-4pt]m2);

    % --- 3. Exchange: Species 1 and Species 1 (jk <-> kj) --- 
    % Site 6.5 is next to the empty site 5.5
    \node[p1] at (6.5, 2.8) {j}; 
    \node[p1] at (7.5, 2.8) {k}; 
    \draw[<->, >=stealth, thick] (6.6, 3.2) to[bend left=45] node[above] {1} (7.4, 3.2); 

    % --- 4. Exchange: Species 1 and 2 (l \bar{0} <-> \bar{0} l) --- 
    % Site 8.5 is left empty to separate the blocks
    \node[p1] at (9.5, 2.8) {l}; 
    \node[p2] at (10.5, 2.8) {$\bar{0}$}; 
    \draw[<->, >=stealth, thick] (9.6, 3.2) to[bend left=45] node[above] {1} (10.4, 3.2); 

    % --- 5. Evaporation: Species 2 pair -> Vacuum (\bar{0}\bar{0} -> 00) --- 
    % Site 11.5 is left empty to separate the blocks
    \node[p2] at (12.5, 2.8) {$\bar{0}$}; 
    \node[p2] at (13.5, 2.8) {$\bar{0}$}; 
    % Upward arrow above the evaporating pair 
    \draw[->, >=stealth, thick] (13.0, 3.4) -- (13.0, 4.1) node[midway, right] {1}; 

    % --- 6. Empty site ---
    % Site 14.5 is left empty to separate the evaporating and condensing pairs

    % --- 7. Condensation: Vacuum -> Species 2 pair (00 -> \bar{0}\bar{0}) --- 
    \node[p2] at (15.5, 2.8) {$\bar{0}$}; 
    \node[p2] at (16.5, 2.8) {$\bar{0}$}; 
    % Downward arrow above the condensing pair 
    \draw[<-, >=stealth, thick] (16.0, 3.4) -- (16.0, 4.1) node[midway, right] {1}; 

    \end{tikzpicture}
    \caption{Illustration of the dynamical rules in the evaporation-condensation interpretation. Particles of species 0 are represented by empty sites and the conjugate species $ \bar{0} $ in grey correspond to evaporating-condensing particles. Singlet particles in black are labeled by $ i,j,k,l \in \llbracket 1,p \rrbracket $.} 
   \vspace{0.5cm}
    \label{fig:evap-cond}
\end{figure}
We first consider the $(p,q)$-SSEP on a periodic lattice. The integrability property of the periodic $(p,q)$-SSEP relies on the Baxterisation \cite{Jones1990} of the set-theoretical solution that gives its spectral parameter dependence. This allows to construct the transfer matrix associated with the process. In the periodic case, the model reaches different uniform stationary distributions in the long-time limit depending on the initial configuration. Non-trivial quantities are conserved by the dynamics. This is studied in detail for the $ (p,1) $-SSEP in Appendix \ref{sec:sec_statio} where the different sectors are fully and explicitly characterised.

In the open boundary case, we find two different types of boundary reservoirs that preserve integrability. They are constructed from solutions of the constant reflection equation and the Baxterisation process. We first present the two different types of boundaries for the $(p,1)$-SSEP. Roughly speaking, the first type of boundary allows the mixing of reactive species with singlet species, but does not distinguish between two reactive species in a pair. On the other hand, the second type of boundaries does not allow mixing between reactive species and singlet species, but instead allows us to distinguish between the two reactive species in a pair. For the $(p,q)$-SSEP, we provide a large class of integrable boundaries which encompasses and generalises the two previous types of boundaries. For the $(p,q)$-SSEP with the first type of boundaries, where the symmetry between reactive species in the same pair is conserved, we are able to calculate physical quantities in the stationary state, such as the mean lattice currents and the densities. This is done by identifying reactive species in a pair, which then allows us to relate the stationary distribution of this open $ (p,q) $-SSEP with the one of an open multi-species SSEP with $p+q$ species.

For the second type of integrable boundaries, the symmetry between reactive species in a pair is broken. Therefore, the physical quantities are sensitive to the difference between two species in the same pair, and thus they cannot all be calculated anymore directly by using the multi-species SSEP. In particular, we define an additional physical quantity called the reaction rate, specific to reactive species, which corresponds, depending on the interpretation, either to a mean vertical condensation current or a mean production rate. These reaction rates are non-vanishing for the second type of boundaries (whereas they are vanishing, by symmetry, for the first type of boundaries). We leave the open problem of calculating the physical quantities, such as the densities, the mean currents and the reaction rates for a future work. 

\medskip

The plan of the paper goes as follows. Section \ref{sec:mark_set} discusses the exclusion Markov processes constructed from set-theoretical solutions of the YBE. The periodic $ (p,1) $-SSEP is defined.

Section \ref{sec:open_(p,1)} first recalls general notions of integrable Markov processes with open boundaries. Then, the $ (p,1) $-SSEP with integrable boundaries $ B^{(1)} $ and $ B^{(2)} $ are defined.

Section \ref{sec:(p,q)-SSEP} introduces the periodic $ (p,q) $-SSEP and discusses its symmetries. The $ (p,q) $-SSEP with a boundary $ B $ that generalises the boundaries $ B^{(1)}, B^{(2)} $ of the $ (p,1) $-SSEP is defined and proved to be integrable.

Section \ref{sec:phys_quant} focuses on physical quantities in the $ (p,q) $-SSEP such as mean currents, densities and reaction rates in the stationary state. They are determined when the form of the boundary $ B $ allows to relate the process with an open multi-species SSEP.

\section{Markov models from set-theoretical solutions}\label{sec:mark_set}

In this section, we introduce a Markov model constructed from a set-theoretical solution of the YBE that is not a Lyubashenko solution. It is integrable and we discuss two possible physical interpretations.

\subsection{Formalism of Markov models on a lattice}\label{ssec:forma_mark}

We briefly recall the basic relevant formalism of periodic exclusion Markov processes.

A Markov exclusion process is defined on a one-dimensional lattice with $ L $ sites. On each site $ i\in \llbracket 1,L \rrbracket $, the local configuration variable $ \tau_{i}$ takes values in a set $X$ of cardinality $ N>0$. The $N$ elements of $X$ represent the $N$ possible states a site $ i $ can take. We refer to these values as particle species. The finite set of configurations of the process is denoted $ \mathfrak{C} $ and corresponds to $ L $-tuples $ \vec \tau=(\tau_{1},\dots,\tau_{L})\in X^L$. 

We let $V$ denote the vector space with basis $ \lvert \tau \rangle $ with $\tau\in X$. It is the state space on one site. The vector space associated to the lattice is $V^{\otimes L}$ and has a basis 
$$\lvert \vec \tau \rangle=\lvert \tau_1 \rangle\otimes \dots\otimes \lvert \tau_L\rangle$$ indexed by configurations $\vec\tau=(\tau_{1},\dots,\tau_{L})\in  X^L$.

The stochastic dynamics of a continuous time Markov process is stored in the Markov matrix
\begin{equation}
M=\sum_{\vec \tau,\vec \tau'\in \mathfrak{C}}m(\vec \tau' \rightarrow \vec \tau)\,\lvert \vec \tau\rangle \langle \vec \tau'\rvert\, ,
\end{equation}
which is acting on the tensor product $V^{\otimes L}$. The entry $ m(\vec \tau' \rightarrow \vec \tau)\geq 0 $, for $ \vec \tau\neq \vec \tau' $, is the transition rate (probability per unit of time) between configurations $ \vec \tau' $ and $ \vec \tau $, and we define 
\begin{equation}\label{coeffdiagM}
m(\vec \tau \rightarrow \vec \tau)=-\sum_{\vec \tau'\in\mathfrak{C},\vec \tau'\neq \vec \tau}m(\vec \tau \rightarrow \vec \tau')\ ,
\end{equation}
so that the sum of the entries in each column of $ M $ is equal to 0.
When the off-diagonal entries are non-negative and the entries of any column add to 0, $ M $ is called Markovian.

Starting from an initial probability distribution $P_0$ on the configuration space, the probability of the system to be in configuration $ \vec \tau $ at continuous time $ t $ is denoted $ P_{t}(\vec \tau) $. The probabilities of all configurations at time $ t $ are stored in the vector $ \lvert P_{t}\rangle = \sum_{\vec \tau\in \mathfrak{C} } P_{t}(\vec \tau) \lvert \vec \tau\rangle$.
The evolution of $\lvert P_{t}\rangle$ is governed by the master equation
\begin{equation}\label{eq:mast_eq}
\mathrm{d}\lvert P_{t}\rangle/\mathrm{dt}=M\lvert P_{t}\rangle .
\end{equation}
For a model with periodic boundary conditions, the Markov matrix has the form
\begin{equation}\label{eq:M_period}
	M=\sum_{i=1}^{L-1}m_{i,i+1}+m_{L,1}\,,
\end{equation}
where $m$ is  the local jump operator acting on two sites and $m_{i,j}$ denotes its action on the two sites $ i $ and $j$. The Markovianity of the local jump operator implies that of $ M $ since this property is preserved by tensoring with the identity and by summing. 

For a model with open boundaries, the Markov matrix has the form
\begin{equation}\label{eq:M_open}
	M=B_{1}+\sum_{i=1}^{L-1}m_{i,i+1}+\bar{B}_{L}\,,
\end{equation}
where $ B_{1} $ and $ \bar{B}_{L} $ denote the actions of Markovian boundary matrices $B$ and $\bar B$ acting on the first and last site respectively. They are interpreted as reservoirs of particles injecting or extracting particles at both ends of the lattice.

\subsection{Set-theoretical solutions of the Yang--Baxter equation}\label{ssec:set_theo}

We briefly recall the basic notions of set-theoretical solutions of the Yang--Baxter equation (YBE) and define the corresponding periodic Markov process.

Let $ \rc $ be a bijection of $ X \times X $. We write
\begin{equation}\label{eq:r_action_gene_set}
	\rc(i,j)=(g_{i}(j),f_{j}(i))\,,
\end{equation}
where $g_i,f_j$ are maps from $X$ to $X$. We say that $\rc$ is a set-theoretical solution of the YBE if
\begin{equation}\label{eq:cYB}
\check r_{12}\,\check r_{23}\,\check r_{12}=\check r_{23}\,\check r_{12}\,\check r_{23}\ ,
\end{equation}
with $ \rc_{12}=\rc \times \mathrm{Id}_{X}: X^{3} \rightarrow X^{3} $ and $ \rc_{23}=\mathrm{Id}_{X} \times \rc: X^{3} \rightarrow X^{3} $. We say that $\rc $ is involutive if $ \rc^{2}=\mathrm{Id}_{X^{2}} $ and non-degenerate if $ g_{i},f_{i} $ are bijections of $ X $ for all $ i\in X $. In the following, we will only consider involutive and non-degenerate solutions.
\begin{remark}\label{rem:Lyub}
For $ g $ a bijection of $ X $, an involutive Lyubashenko solution of the YBE corresponds to $ g_{i}=g $ and $ f_{i}=g^{-1}$ for all $ i\in X $. Markov models corresponding to involutive Lyubashenko solutions are studied in \cite{ragoucy2026}.
\end{remark}
To a set-theoretical solution $\rc$ corresponds a matrix, also denoted $ \rc $, acting on $V\otimes V$ by
\begin{equation}\label{eq:r_action_gene_vec}
	\rc \lvert i,j \rangle = \lvert g_{i}(j), f_{j}(i) \rangle\,.
\end{equation}
From the matrix $\rc$, we define the local jump operator
\begin{equation}\label{eq:loc_jump_r}
m=\check r-\mathrm{Id}_{V\otimes V}\,,
\end{equation}
which is Markovian because $ \rc $ is a permutation matrix (that is, $\check r$ has exactly one entry equal to $1$ in each column and row and all other entries are $0$). The Markov matrix $M$ of the associated periodic model is defined from $ m $ as in Equation \eqref{eq:M_period}. 

The matrix $\rc$ is symmetric, because orthogonal and involutive, and therefore so is the Markov matrix $M$. In other words, the transition rates satisfy
\begin{equation}\label{eq:sym_rate}
\forall \vec \tau,\vec \tau'\in\mathfrak{C}\,,\ \ m(\vec \tau \rightarrow \vec \tau')=m(\vec \tau' \rightarrow \vec \tau)\,.
\end{equation}

\subsection{Definition of the periodic $(p,1)$-SSEP}\label{ssec:def_(p,1)}

Let $p\geq 0$. We now define the periodic model that we call periodic $ (p,1) $-SSEP. We take
\begin{equation}\label{eq:(p,1)_SSEP_label}
	X=\{0, \bar{0}, 1,\dots, p\}\,.
\end{equation}
We call the pair $\{0,\bar 0\}$ a pair of {\it{reactive particle species}}. The remaining states correspond to $p$ usual particle species which behave as multi-species SSEP particles. We refer to them as ``singlet'' particle species. The set-theoretical matrix $\rc^{(p,1)}$ is given by:
\begin{equation}\label{def:rc(p,1)}
\rc^{(p,1)}:\ \begin{array}{l}  \lvert 00\rangle \mapsto \lvert \bar 0 \bar 0\rangle\,, \\
\lvert \bar 0 \bar 0\rangle \mapsto \lvert 0 0\rangle\,, \\
\lvert 0\bar 0\rangle \mapsto \lvert 0 \bar 0\rangle\,, \\
\lvert \bar 0 0\rangle \mapsto \lvert \bar 0 0\rangle\,,\end{array}\qquad \begin{array}{l}  \lvert 0i\rangle \mapsto \lvert i 0\rangle\,, \\
\lvert i 0\rangle \mapsto \lvert 0 i\rangle\,, \\
\lvert i \bar 0\rangle \mapsto \lvert \bar 0 i\rangle\,, \\
\lvert \bar 0 i\rangle \mapsto \lvert i \bar 0 \rangle\,, \end{array}\qquad \lvert ij\rangle \mapsto \lvert ji\rangle\ ,
\end{equation}
where $i,j\in\llbracket 1,p \rrbracket$. Equivalently, it is defined by
\begin{equation}\label{eq:r_(p,1)_set_1}
\rc^{(p,1)} \lvert i,j \rangle = \lvert f_{i}(j), f_{j}(i) \rangle\,,	
\end{equation}
where $ f_{0}=f_{\bar{0}} $ is the transposition $ (0,\bar{0}) $ and $ f_{i} $ is the identity for $ i\in \llbracket 1,p \rrbracket$.

On the subspace with basis $(\lvert 0,0 \rangle, \lvert 0,\bar{0} \rangle, \lvert \bar{0},0 \rangle, \lvert \bar{0},\bar{0} \rangle)$, the matrix $\rc^{(p,1)}$ acts as follows:
\begin{equation}\label{eq:R_(1)_def}
\begin{pmatrix}
	\cdot & \cdot & \cdot & 1\\
\cdot & 1 & \cdot & \cdot\\
\cdot & \cdot & 1 & \cdot\\
1 & \cdot & \cdot & \cdot
\end{pmatrix}\,,
\end{equation}
while on the remaining basis vectors, the matrix $\rc^{(p,1)}$ acts simply as the permutation.

\begin{proposition}\label{prop:YBE_r_(p,1)}
The matrix $ \rc^{(p,1)} $ is solution of the YBE.
\end{proposition}
\begin{proof}
The proof is a straightforward verification, that will be done in a more general case in Proposition \ref{prop:YBE_r_(p,q)}.
\end{proof}
The matrix $ \rc^{(p,1)} $ is clearly non-degenerate, involutive and symmetric. Also, if $p\neq0$, this is not a Lyubashenko solution; cf. Remark \ref{rem:Lyub}.

The Markov matrix $M^{(p,1)}$ is constructed from $\rc^{(p,1)}$ using Equations \eqref{eq:M_period} and \eqref{eq:loc_jump_r}. The allowed local configuration changes are for $ i,j\in \llbracket 1,p \rrbracket $
\begin{equation}\label{eq:loc_move}
	ij \longleftrightarrow ji,\quad 0i \longleftrightarrow i0,\quad \bar{0}i \longleftrightarrow i\bar{0},\quad 00 \longleftrightarrow \bar{0}\bar{0}\,.
\end{equation}
For instance, in terms of transition rates, $ 0i \leftrightarrow i0 $ means  that $ m(0i \rightarrow i0)= m(i0 \rightarrow 0i)=1 $.

\begin{example}\label{ex:(1,1)_SSEP}
The matrix $ \rc^{(1,1)} $ of the $ (1,1) $-SSEP is given in the basis 
\begin{equation}\label{eq:B_(1,1)}
	( \lvert 0,0 \rangle, \lvert 0, \bar{0} \rangle, \lvert \bar{0}, 0 \rangle, \lvert \bar{0}, \bar{0} \rangle, \lvert 0, 1 \rangle, \lvert 1,0 \rangle, \lvert \bar{0},1 \rangle, \lvert 1, \bar{0} \rangle, \lvert 1,1 \rangle)\,,
\end{equation}
by
\begin{equation}\label{eq:r_(1,1)}
	\rc^{(1,1)}=
\begin{pmatrix}
	\cdot &\cdot&\cdot &1&\cdot &\cdot&\cdot &\cdot&\cdot\\
\cdot &1&\cdot &\cdot&\cdot &\cdot&\cdot &\cdot&\cdot\\
\cdot &\cdot&1 &\cdot&\cdot &\cdot&\cdot &\cdot&\cdot\\
1 &\cdot&\cdot &\cdot&\cdot &\cdot&\cdot &\cdot&\cdot\\
\cdot &\cdot&\cdot &\cdot&\cdot &1&\cdot &\cdot&\cdot\\
\cdot &\cdot&\cdot &\cdot&1 &\cdot&\cdot &\cdot&\cdot\\
\cdot &\cdot&\cdot &\cdot&\cdot &\cdot&\cdot &1&\cdot\\
\cdot &\cdot&\cdot &\cdot&\cdot &\cdot&1 &\cdot&\cdot\\
 \cdot&\cdot&\cdot &\cdot&\cdot &\cdot&\cdot &\cdot&1\\
\end{pmatrix}\, .
\end{equation}
It can be shown that the matrix $ \rc^{(1,1)} $ comes from a brace. Braces are sets with two group structures and a specific distributive property used to generate any non-degenerate involutive solution of the YBE \cite{Cedo2012}. The brace leading to $ \rc^{(1,1)} $ has its additive group isomorphic to $ \mathbb{Z}/2\mathbb{Z} \times \mathbb{Z}/2\mathbb{Z} $ and its multiplicative group isomorphic to $ \mathbb{Z}/4\mathbb{Z} $. 
\end{example}

\subsection{Integrability of the periodic $(p,1)$-SSEP}\label{ssec:int_(p,1)}

Along the exact same lines as in \cite[Sec.2.3]{ragoucy2026}, we briefly prove the integrability of the periodic $ (p,1) $-SSEP. The $ (p,1) $-SSEP is integrable if its  Markov matrix belongs to a set of commuting operators generated by a transfer matrix. 

It is well known (see e.g. \cite{faddeev1996,FST}) that the operator
\begin{equation}\label{eq:R_u}
	\rc(u)=\frac{u\, \rc+ \mathrm{Id}_{V \otimes V}}{u+1}\,,
\end{equation}
satisfies the YBE with spectral parameter
\begin{equation}\label{eq:braid_spec}
\check r_{12}(u)\,\check r_{23}(u+v)\,\check r_{12}(v)=\check r_{23}(v)\,\check r_{12}(u+v)\,\check r_{23}(u)\,,
\end{equation}
when $ \rc $ satisfies the YBE without spectral parameter and is involutive.

For $r(u)=P\rc(u)$ with $P$ the permutation operator, it is also well known that
\begin{equation}\label{eq:trans_period}
	t(u)=\mathrm{tr}_{0}\, r_{0,L}(u)\,r_{0,L-1}(u)\cdots r_{0,1}(u)\, ,
\end{equation}
is the transfer matrix satisfying
\begin{equation}\label{eq:trans_period_2}
[t(u)\,,\,t(v)]=0\,,
\end{equation}
so that upon expansion in $u$, $t(u)$ indeed generates a set of commuting operators. 

A straightforward and classical calculation (see e.g. \cite{vanicat2017}) shows that 
$$t(0)^{-1}t'(0)=\sum_{i=1}^{L-1}m_{i,i+1}+m_{L,1} \,,\ \ \ \text{where $m=\check r-\text{Id}$.}$$ We recover in this way the Markov matrix of the periodic model associated to $\check r$ introduced in Section \ref{ssec:set_theo}, and therefore the model is integrable. In particular, the periodic model of Section \ref{ssec:def_(p,1)}, constructed from
the set-theoretical solution $\check r^{(p,1)}$ is integrable.

\subsection{Interpretations of the periodic $(p,1)$-SSEP}\label{ssec:interpret_(p,1)}

We now give two possible physical interpretations of the $(p,1)$-SSEP. In both cases, we have a multi-SSEP with $p$ species, to which is added a pair of reactive species which behaves differently.

\subsubsection{Evaporation/condensation of impurities}

Let the particle species $ 0 $ represent an empty site. In that case, local configuration changes in \eqref{eq:loc_move} can be interpreted as follows. The singlet particles of species $ i\in \llbracket 1,p \rrbracket $ can freely move across empty sites or exchange with particles of species $ \bar{0} $ or $j\in \llbracket 1,p \rrbracket$, $j\neq i$. On the contrary, particles of species $\bar{0}$ cannot move to empty sites by themselves; they can only be displaced on the lattice through the motion of the singlet particles. They can therefore be regarded as impurities. The local change $ 00 \leftrightarrow \bar{0}\bar{0} $ can be interpreted as pairwise evaporation and condensation of impurities $ \bar{0} $. 
This interpretation is illustrated in Figure \ref{fig:evap-cond} in the introduction.

\begin{remark}\label{rem:DiSSEP}
A process with a similar interpretation as the $ (p,1) $-SSEP for $ p=0 $ is the Dissipative Symmetric Exclusion Process (DiSSEP) introduced in \cite{vanicat2016}. It differs however because the evaporating-condensing particles cannot move across empty sites in the $ (0,1) $-SSEP but can in the DiSSEP. Moreover, in the $ (0,1) $-SSEP, the evaporation-condensation rate is equal to unity but it can be any positive parameter in the DiSSEP.
\end{remark}

\subsubsection{Transformation  of reactive particles}

Let each state (including $0$) represent a particle so that the lattice is filled with particles. In that case, every singlet particle can freely move by exchanging its position with other singlet or reactive particles. The particles of species $ 0 $ and $ \bar{0} $ cannot overtake each other. The local change $00 \leftrightarrow \bar{0}\bar{0} $ can be interpreted as a transformation of pairs of particle species $ 0 $ into their conjugate $ \bar{0} $ and vice versa. More precisely, the transformation can be a chemical reaction between two particles of the same species, which change into two particles of the other reactive species of the same pair. This can also describe a change of conformation of the same molecule, or a transition between two degenerate states of an atom for example.

\begin{figure}[H]
    \centering
    % \resizebox{\columnwidth}{!}{
    \begin{tikzpicture}[
        scale=0.8, transform shape,
        % Species 0: Gray fill, thick border, black text
        p0/.style={circle, draw=black, thick, fill=gray!20, text=black, minimum size=14pt, inner sep=0pt, font=\small\bfseries},
        % Species e,f,g... (Singlets): Black fill, white text
        p1/.style={circle, draw=black, fill=black, text=white, minimum size=14pt, inner sep=0pt, font=\small\bfseries},
        % Species 2 (bar0): Gray fill, thick border, black text
        p2/.style={circle, draw=black, thick, fill=gray!20, text=black, minimum size=14pt, inner sep=0pt, font=\small\bfseries}
    ]

    % =========================================================
    % MAIN LATTICE
    % =========================================================
    % Periodic extensions and main lattice line
    \node at (-0.8, 2.6) {$\cdots$};
    \draw[dashed, thick] (-0.5, 2.5) -- (0, 2.5);
    \draw (0,2.5) -- (12,2.5);
    \draw[dashed, thick] (12, 2.5) -- (12.5, 2.5);
    \node at (12.8, 2.6) {$\cdots$};

    % Lattice ticks
    \foreach \x in {0,...,12} \draw (\x,2.5) -- (\x,2.7);

    % --- PARTICLES ---
    % 1. Singlet-Singlet
    \node[p1] at (0.5, 2.8) {i};
    \node[p1] at (1.5, 2.8) {j};
    
    % 2. Singlet-Reactive (0)
    \node[p1] at (2.5, 2.8) {k};
    \node[p0] at (3.5, 2.8) {0};
    
    % 3. Singlet-Reactive (\bar{0})
    \node[p1] at (4.5, 2.8) {l};
    \node[p2] at (5.5, 2.8) {$\bar{0}$};
    
    % 4. Blocked Exchange
    \node[p0] at (6.5, 2.8) {0};
    \node[p2] at (7.5, 2.8) {$\bar{0}$};

    % 5. Creation-Annihilation 1
    \node[p0] at (8.5, 2.8) {0};
    \node[p0] at (9.5, 2.8) {0};
    
    % 6. Creation-Annihilation 2
    \node[p2] at (10.5, 2.8) {$\bar{0}$};
    \node[p2] at (11.5, 2.8) {$\bar{0}$};

    % =========================================================
    % ACTIONS: Possible Local Changes
    % =========================================================
    % --- Valid Exchanges (Rate 1) ---
    % Singlet-Singlet
    \draw[<->, >=stealth, thick] (0.6, 3.2) to[bend left=45] node[above] {1} (1.4, 3.2);
    % Singlet-Reactive (0)
    \draw[<->, >=stealth, thick] (2.6, 3.2) to[bend left=45] node[above] {1} (3.4, 3.2);
    % Singlet-Reactive (\bar{0})
    \draw[<->, >=stealth, thick] (4.6, 3.2) to[bend left=45] node[above] {1} (5.4, 3.2);

    % --- Blocked Exchanges (Crossed out) ---
    % 0 \bar{0} block (Black cross)
    \draw[<->, >=stealth, thick] (6.6, 3.2) to[bend left=45] coordinate[midway] (m1) (7.4, 3.2);
    \draw[thick] ([xshift=-4pt,yshift=-4pt]m1) -- ([xshift=4pt,yshift=4pt]m1);
    \draw[thick] ([xshift=-4pt,yshift=4pt]m1) -- ([xshift=4pt,yshift=-4pt]m1);

    % --- Creation-Annihilation Transformations ---
    % Upper braces grouping original particles
    \draw[decorate, decoration={brace, amplitude=4pt}, thick] (8.2, 3.2) -- (9.8, 3.2);
    \draw[decorate, decoration={brace, amplitude=4pt}, thick] (10.2, 3.2) -- (11.8, 3.2);

    % Vertical short up arrows (shifted up to create vertical space above braces)
    \draw[->, >=stealth, thick] (9.0, 3.5) -- (9.0, 3.9) node[midway, right] {1};
    \draw[->, >=stealth, thick] (11.0, 3.5) -- (11.0, 3.9) node[midway, right] {1};

    % Upper braces opening upward (mirror points it downward to cup the upper particles)
    \draw[decorate, decoration={brace, amplitude=4pt, mirror}, thick] (8.2, 4.2) -- (9.8, 4.2);
    \draw[decorate, decoration={brace, amplitude=4pt, mirror}, thick] (10.2, 4.2) -- (11.8, 4.2);

    % Transformed particles on top
    \node[p2] at (8.5, 4.6) {$\bar{0}$};
    \node[p2] at (9.5, 4.6) {$\bar{0}$};
    \node[p0] at (10.5, 4.6) {0};
    \node[p0] at (11.5, 4.6) {0};

    \end{tikzpicture}
    \vspace{0.5cm} % Adds a little breathing room before the caption
    \caption{Illustration of the dynamical rules in the transformation interpretation. Singlet species are in black, labeled by $i,j,k,l\in \llbracket 1,p \rrbracket $, and reactive particle species are in grey, labeled by $ 0 $ and $ \bar{0} $.}
    \label{fig:mutation_process}
\end{figure}

\begin{remark}\label{rem:interp_(p,1)}
In the last transformation interpretation, we could also choose any of the singlet species to represent an empty site. In that case, all particles can move across empty sites. For example, for $p=1$, then $1$ represents an empty site and we have on the lattice the species $0$ and $\bar 0$ with their reactive properties, together with empty sites across which both $0$ and $\bar 0$ can move.
\end{remark}

\begin{remark}\label{rem:reac_diff}
The diffusion-reaction process in \cite{Schutz1994} has a similar interpretation as the $ (1,1) $-SSEP in which the species 1 represents empty sites (see Remark \ref{rem:interp_(p,1)}). However, the reactive particles $ A,B $ of \cite{Schutz1994} do not transform into each other pairwise, but individually, when they jump to an empty neighbouring site.
\end{remark}

\section{Integrable boundaries}\label{sec:open_(p,1)}

In this section, after recalling the construction of the transfer matrix for models with open boundaries, we provide two different solutions of the reflection equation associated to the matrix $\check r$ of the $ (p,1) $-SSEP. This allows to explicitly describe two types of integrable boundaries on each end of the lattice for the $ (p,1) $-SSEP, leading to three inequivalent  processes. In the following, we focus on the two processes where both ends of the lattice are subject to the same type of boundary condition.

\subsection{Integrable models with open boundaries}\label{sec:int_bound}

In the following, we recall how to construct a Markov matrix of the form \eqref{eq:M_open} that belongs to a set of commuting operators coming from a transfer matrix whose expression takes into account the boundaries of the process. We refer to \cite{Sklyanin1988} or \cite{vanicat2017} for more details on the construction.

\subsubsection{Baxterisation of involutive set-theoretical solutions of the YBE}

Recall that we consider a set-theoretical solution of the YBE written as
$$ \rc \lvert i,j \rangle = \lvert g_{i}(j), f_{j}(i) \rangle \ \ \ \ \ \forall i,j\in X\,,$$
where $f_j$ and $g_i$ are bijections of the set $X$ indexing a basis of our vector space $V$. It satisfies the YBE $\check r_{12}\check r_{23}\check r_{12}=\check r_{23}\check r_{12}\check r_{23}$ without spectral parameter, and we assume moreover that it is involutive, namely that $\check r^2=\text{Id}$.

Then, as recalled in Section \ref{ssec:int_(p,1)}, the operator
\begin{equation}\label{eq:R_u2}
	\rc(u)=\frac{u\, \rc+ \mathrm{Id}_{V \otimes V}}{u+1}\,,
\end{equation}
satisfies the YBE with spectral parameter
\begin{equation}\label{eq:braid_spec2}
\check r_{12}(u)\,\check r_{23}(u+v)\,\check r_{12}(v)=\check r_{23}(v)\,\check r_{12}(u+v)\,\check r_{23}(u)\,.
\end{equation}
We say that  $ \rc(u) $  is a Baxterisation of $ \rc $. 

In the proposition below, $(\cdot)^{t_1}$ denotes the partial transpose of an operator acting on $V\otimes V$. It is defined by linearly extending to an arbitrary operator the definition $(A\otimes B)^{t_1}=A^t\otimes B$. The following properties satisfied by $\rc(u)$ ensure the commutativity of the transfer matrix in the open case (see below).
\begin{proposition}\label{prop_condintR}
The Baxterisation (\ref{eq:R_u2}) satisfies
\[\check r(u)^{-1}=\check r(-u)\ \ \ \ \text{and}\ \ \ \ \ r_{12}(u)^{t_{1}}\ \text{is invertible}\,.\]
\end{proposition}
\begin{proof}
The formula $\check r(u)\check r(-u)=\text{Id}$ follows immediately using $\check r^2=\text{Id}$. 

To show the invertibility of $r_{12}(u)^{t_{1}}$, we have to prove that
\begin{equation}\label{eq:r_(p,q)_partial_t}
	\frac{r_{12}^{t_{1}}u+P_{12}^{t_{1}}}{1+u}\,,
\end{equation}
is invertible. Assuming for a moment that $r_{12}^{t_{1}}$ is invertible, we have
\begin{align}\label{eq:det_r_partial_(p,q)}
	\mathrm{det}(r_{12}^{t_{1}}u+P_{12}^{t_{1}})&=\mathrm{det}(r_{12}^{t_{1}}u)\mathrm{det}(\mathrm{Id}+u^{-1}(r_{12}^{t_{1}})^{-1}P_{12}^{t_{1}})\\
					       &=u^{|X|^{2}}\mathrm{det}(r_{12}^{t_{1}})(1+\mathcal{O}(u^{-1}))\,.
\end{align}
This proves that $ r_{12}(u)^{t_{1}} $ is invertible.

It remains to prove that $ r_{12}^{t_{1}} $ is invertible. Explicitly, starting from \eqref{eq:r_action_gene_vec}, the operator $ r_{12}=P_{12}\,\check r_{12}$ is written as
\begin{equation}\label{eq:r_(p,q)_compo}
	r_{12}=\sum_{i,j,k,l\in X}r_{kl}^{ij} \lvert k,l \rangle \langle i,j \rvert\,,\ \ \ \ \text{with $r_{kl}^{ij}=\delta_{k,f_j(i)}\delta_{l,g_i(j)}$\,.}   
\end{equation}
We have 
\begin{equation}\label{eq:rt_(p,q)_compo}
r_{12}^{t_{1}}=\sum_{i,j,k,l\in X}r_{kl}^{ij} \lvert i,l \rangle \langle k,j \rvert\,,
\end{equation}
so the components of the column $ (k,j) $ of $ (r_{12})^{t_{1}} $ are the elements $ r_{kl}^{ij} $ with $ i,l\in X$ and the components of the row $ (i,l) $ are the elements $ r_{kl}^{ij} $ with $ k,j\in X$.
We now show that a row has exactly one non-zero component. For $ (i,l) $ fixed, $ r_{kl}^{ij}=1 $ if and only if $ k=f_{j}(i)$ and $ l=g_{i}(j)$. The map $g_i$ is bijective, therefore there is a unique element denoted $j^*$ such that $l=g_{i}(j^{*})$. There is then a unique element denoted $k^*$ such that $k^{*}=f_{j^{*}}(i)$ so $ r_{kl}^{ij}=1 $ for row $ (i,l) $ only at column $ (k^{*},j^{*}) $.
We prove that a column has exactly one non-zero component in a similar way. This concludes the proof.
\end{proof}

\begin{remark}
It is also straightforward to check that $\check r_{21}(u)=\check r_{12}(u)$ if and only if $g_i=f_i$ for all $i\in X$. This condition is satisfied by all solutions considered in this paper.   
\end{remark}

\subsubsection{Baxterisation for the reflection equation}
We keep the operator $ \rc(u) $ of the preceding subsection, see \eqref{eq:R_u2}, and we look for solutions $k(u)$ of the parameter-dependent reflection equation
\begin{equation}\label{eq:re_para}
\rc_{12}(u-v)k_{1}(u)\rc_{12}(u+v)k_{1}(v)=k_{1}(v)\rc_{12}(u+v)k_{1}(u)\rc_{12}(u-v)\,,
\end{equation}
where here $k(u)$ is an endomorphism of $V$. 

To do so, we start from an endomorphism $k$ on $V$ solution of the constant reflection equation
\begin{equation}\label{eq:cRE}
\rc_{12}k_{1}\rc_{12}k_{1}=k_{1}\rc_{12}k_{1}\rc_{12}\ .
\end{equation}
We are going to use several times the following construction \cite{Doikou_2021_2}
\begin{proposition}\label{prop:baxt_bound}
Let $k$ be a solution of the constant reflection equation \eqref{eq:cRE} such that
\begin{equation}\label{eq:k_baxt_cond}
	k^{2}=a_{1}k+a_{0}\,,\ \ \ \ \ \text{for some $a_0,a_1\in\mathbb{C}$\,.}
\end{equation}
Then, for any scalar function $f(u)$ and any number $c\in\mathbb{C}$, we have that
\begin{equation}\label{eq:k_baxt}
	k(u)=f(u)\left(k+\frac{c-a_{1}u}{2u}\mathrm{Id}\right)\,,
\end{equation}
is a solution of the parameter dependent reflection equation \eqref{eq:re_para}\ .  
\end{proposition}
We say that  $ k(u) $ is a Baxterisation of $ k $.

\subsubsection{Transfer matrix}

We are ready to construct the transfer matrix out of the following ingredients:
\begin{itemize}
    \item A solution $\rc(u)$ of the parameter dependent YBE equation satisfying the invertibility conditions of Proposition \ref{prop_condintR}.
    \item A solution $ k(u) $ of the parameter dependent reflection equation \eqref{eq:re_para}.
    \item A solution $\bar k(u)$ of the parameter dependent reversed reflection equation, obtained from the reflection equation \eqref{eq:re_para} by replacing $ \rc_{12}(u) $  by $\check r_{21}(-u)$.
\end{itemize}
\begin{remark}\label{rem_reversedRE}
In the case where $\check r_{21}(u)=\check r_{12}(u)$ (which is always true in this paper) we have that,  for any solution $k(u)$ of the reflection equation, $\bar k(u)=k(-u)$ is clearly a  solution of the reversed reflection equation.
\end{remark}
From $\bar k(u)$, using the invertibility condition of Proposition \ref{prop_condintR}, we define:
\begin{equation}\label{eq:k_bar_k_tilde}
	\tilde{k}_{1}(u)=\mathrm{tr}_{0}\left(\bar{k}_{0}(-u)\left(\left(r_{01}^{t_{1}}(2u)\right)^{-1}\right)^{t_{1}}P_{01}\right)\,.
\end{equation}
The endomorphism $ \tilde{k}(u) $ of $ V $ is solution of the parameter dependent dual reflection equation (see e.g. \cite{vanicat2017_2})
\begin{equation}\label{eq:dual_re}
\tilde{k}_{2}(v)\left(r_{21}^{t_{1}}(u+v)^{-1}\right)^{t_{1}}\tilde{k}_{1}(u)r_{21}(v-u)=r_{12}(v-u)\tilde{k}_{1}(u)\left(r_{12}^{t_{2}}(u+v)^{-1}\right)^{t_{2}}\tilde{k}_{2}(v)\,.
\end{equation}
Now we can define the transfer matrix (see \cite{Sklyanin1988}) by:
\begin{equation}\label{eq:trans}
t(u)=\mathrm{tr}_{0}\left( \tilde{k}_{0}(u)r_{0L}(u)\dots r_{01}(u)k_{0}(u)r_{10}(u)\dots r_{L0}(u)\right)\,.
\end{equation}
and it satisfies the fundamental commutativity property:
\begin{equation}\label{eq:trans_open}
[t(u),t(v)]=0\,.
\end{equation}

\subsubsection{Markov matrix}

To extract a Markov matrix from the transfer matrix, we assume the following  regularity conditions:
\begin{equation}\label{eq:r_reg}
	r(0)=P\ \ \ \ \ \ \text{and}\ \ \ \ \ \  k(0)=\mathrm{Id}\,,\quad \bar{k}(0)=\mathrm{Id}\,.
\end{equation}
Using the following relation between $ \bar{k}(u) $ and $ \tilde{k}(u) $ (see e.g. \cite{vanicat2017}) 
\begin{equation}\label{eq:k_tilde_k_bar}
\bar{k}_{1}(u)=\mathrm{tr}_{0}\left(\tilde{k}_{0}(-u)r_{01}(-2u)P_{01}\right)\,,
\end{equation} and setting
\begin{equation}\label{eq:int_mark_op}
	B:=\frac{1}{2}k'(0),\qquad m:=\rc'(0),\qquad \bar{B}:=-\frac{1}{2}\,\bar{k}'(0)\,,
\end{equation}
one can prove by a straightforward computation that 
\begin{equation}\label{eq:M_int_open}
\frac{1}{2}\,t'(0)=B_{1}+\sum_{i=1}^{L-1}m_{i,i+1}+\bar{B}_{L}\,.
\end{equation}
The operator found in \eqref{eq:M_int_open} has exactly the form of a Markov matrix for a model with open boundaries (see Section \ref{ssec:forma_mark}).  The Markov matrix $ M $ obtained from the procedure \eqref{eq:M_int_open} is said to define an integrable Markov process with open boundaries.

\subsection{Integrable boundaries for the $(p,1)$-SSEP}\label{ssec:bound_(p,1)}

Recall that for the $ (p,1) $-SSEP, the set $X$ is $\{0,\bar 0,1,\dots,p\}$ and therefore the canonical basis of $ V $ is
\begin{equation}
	( \lvert 0 \rangle, \lvert \bar{0} \rangle,  \lvert 1 \rangle,\dots, \lvert p \rangle)\,.
\end{equation}
The periodic model is entirely defined by the matrix $ \check r^{(p,1)}$ of Section \ref{ssec:def_(p,1)}. To construct integrable models with open boundaries, our task is to find solutions to the reflection equation, which can be done using the Baxterisation procedure mentioned previously, and then to derive the corresponding boundary matrices $B$ and $\bar B$ suitable for a Markov model.

\begin{remark}
Note that there is a notion of set-theoretical solution of the reflection equation \cite{Caudrelier2013}, \cite{Doikou_2021_2}  but the solutions we consider below are not of this form.
\end{remark}

\subsubsection{Definition of the boundaries $ B^{(1)} $}\label{sssec:def_(p,1)_B_(1)}

Let
\begin{equation}\label{eq:k_(p,1)}
k^{(1)} = \frac{2}{\lambda}\left(
\begin{array}{ccccc}
\gamma & \gamma & \dots & \dots & \gamma\\
\gamma & \gamma & \dots & \dots & \gamma \\
\alpha_1 & \alpha_1 & \dots & \dots & \alpha_1 \\
\alpha_2 & \alpha_2 & \dots & \dots & \alpha_2 \\
\vdots & \vdots & & & \vdots \\
\alpha_p & \alpha_p & \dots & \dots & \alpha_p
\end{array}
\right)\,\ \ \ \ \ \text{with $2\gamma+\sum_{i=1}^p\alpha_i=1$.}\end{equation}

\begin{proposition}
The matrix $k^{(1)}$ satisfies the constant reflection equation with the matrix $\check r^{(p,1)}$:
\[\rc^{(p,1)}_{12}k^{(1)}_{1}\rc^{(p,1)}_{12}k^{(1)}_{1}=
k^{(1)}_{1}\rc^{(p,1)}_{12}k^{(1)}_{1}\rc^{(p,1)}_{12}\ .\]
\end{proposition}
\begin{proof}
This is a special case of Proposition \ref{prop_reflection_(p,q)} proved in Section \ref{sec:(p,q)-SSEP}.
\end{proof}

The matrix $k^{(1)}$ satisfies $(k^{(1)})^2=\frac{2}{\lambda}k^{(1)}$ and therefore we know how to baxterise it thanks to Proposition \ref{prop:baxt_bound}. We will use
\begin{equation}\label{eq:k_u_(p,1)}
	k^{(1)}(u)=\Bigl(1-\frac{u}{\lambda}\Bigr)\Bigl(uk^{(1)}+\Bigl(1-\frac{u}{\lambda}\Bigr)\mathrm{Id}\Bigr)\,,
\end{equation}
This provides a solution of the parameter-dependent reflection equation.

To get a solution of the reversed reflection equation, we start with the solution of the constant reflection equation 
\begin{equation}\label{eq:kb_(p,1)}
 \bar{k}^{(1)} = \frac{2}{\mu}\left(
\begin{array}{ccccc}
\beta & \beta & \dots & \dots & \beta\\
\beta & \beta & \dots & \dots & \beta \\
\delta_1 & \delta_1 & \dots & \dots & \delta_1 \\
\delta_2 & \delta_2 & \dots & \dots & \delta_2 \\
\vdots & \vdots &  &  & \vdots \\
\delta_p & \delta_p & \dots & \dots & \delta_p
\end{array}
\right)\ \ \ \ \ \text{with $2\beta+\sum_{i=1}^p\delta_i=1$.}
\end{equation}
and we baxterise it with spectral parameters $-u$ instead of $u$ (see Remark \ref{rem_reversedRE}). We will use:
\begin{equation}\label{eq:k_u_bar_(p,1)}
	\bar{k}^{(1)}(u)=-\Bigl(1+\frac{u}{\mu}\Bigr)\Bigl(u\bar{k}^{(1)}-\Bigl(1+\frac{u}{\mu}\Bigr)\mathrm{Id}\Bigr)\,,
\end{equation}

The left and right boundary matrices are given by 
\begin{equation}\label{eq:B_(p,1)_k}
	B^{(1)}=\frac{1}{2}(k^{(1)})'(0)=\frac{1}{2}\left(k^{(1)}-\frac{2}{\lambda} \mathrm{Id}\right)\,,\qquad\ \ 
	\bar{B}^{(1)}=-\frac{1}{2}(\bar{k}^{(1)})'(0)=\frac{1}{2}\left(\bar{k}^{(1)}-\frac{2}{\mu}\mathrm{Id}\right)\,.
\end{equation}
In matrix form, we find:
\begin{align}\label{eq:B_(p,1)}
	B^{(1)} = \frac{1}{\lambda} \left(
\begin{array}{cccccc}
\gamma & \gamma & \dots & \dots & \gamma \\
\gamma & \gamma & \dots & \dots & \gamma \\
\alpha_1 & \alpha_1 & \dots & \dots & \alpha_1 \\
\alpha_2 & \alpha_2 & \dots & \dots & \alpha_2 \\
\vdots & \vdots & & & \vdots \\
\alpha_p & \alpha_p & \dots & \dots & \alpha_p 
\end{array}
\right)-\frac{1}{\lambda}\mathrm{Id}\,,
\end{align}
\begin{align}\label{eq:B_(p,1)_bar}
\bar B^{(1)} = \frac{1}{\mu} \left(
\begin{array}{cccccc}
\beta & \beta & \dots & \dots & \beta \\
\beta & \beta & \dots & \dots & \beta \\
\delta_1 & \delta_1 & \dots & \dots & \delta_1 \\
\delta_2 & \delta_2 & \dots & \dots & \delta_2 \\
\vdots & \vdots & & & \vdots \\
\delta_p & \delta_p & \dots & \dots & \delta_p 
\end{array}
\right)-\frac{1}{\mu}\mathrm{Id}\,.
\end{align}
We choose the parameters to be positive: $ \lambda,\mu,\gamma,\beta,\alpha_{1},\ldots,\alpha_{p},\delta_1,\ldots,\delta_p>0 $. The constraints
$\sum_{i=1}^{p} \alpha_{i}+ 2\gamma=\sum_{i=1}^{p} \delta_{i}+ 2\beta=1$ ensure that we get Markov matrices.

Recall that from $ B^{(1)} $ and $\bar B^{(1)}$, we read the transition rates at the boundaries. For example, we read from $ B^{(1)} $ the following transition rates at the left boundary:
\begin{equation}\label{eq:rates_B_(p,1)}
	0 \overset{\gamma/\lambda}{\longleftrightarrow}\bar{0},\quad  i \xrightleftharpoons[\alpha_{i}/\lambda]{\gamma/\lambda}0,\quad i \xrightleftharpoons[\alpha_{i}/\lambda]{\gamma/\lambda}\bar{0},\quad i \xrightleftharpoons[\alpha_{i}/\lambda]{\alpha_{j}/\lambda}j\,,
\end{equation}
where $ i $ and $ j $ are different singlet particles.
\begin{figure}[H]
    \centering
    \begin{tikzpicture}[
        scale=0.85, transform shape,
        % Particle i, j: Black fill, white text
        p1/.style={circle, draw=black, fill=black, text=white, minimum size=14pt, inner sep=0pt, font=\bfseries},
        % Particle \bar{0}: Gray fill, thick border, black text
        p2/.style={circle, draw=black, thick, fill=gray!20, minimum size=14pt, inner sep=0pt, font=\small\bfseries}
    ]

    % Define a macro for the first site
    % It draws the left wall, bottom floor, right wall, AND a horizontal continuation
    \def\site#1#2{
        \draw (#1-0.4, #2+0.15) -- (#1-0.4, #2) -- (#1+0.4, #2) -- (#1+0.4, #2+0.15);
        \draw (#1+0.4, #2) -- (#1+0.8, #2); % Extension to the right
    }

    % =========================================================
    % TIME LABELS
    % =========================================================
    \node[anchor=east, font=\bfseries] at (-2.5, 3.0) {Time $t$};
    \node[anchor=east, font=\bfseries] at (-2.5, 0.0) {Time $t+dt$};

    % =========================================================
    % TREE 1: Start with 0 (Empty Site)
    % =========================================================
    % Time t
    \site{0}{3}

    % Time t+dt
    \site{-1.2}{0}
    \node[p2, anchor=south] at (-1.2, 0.03) {$\bar{0}$};

    \site{1.2}{0}
    \node[p1, anchor=south] at (1.2, 0.03) {$i$};

    % Arrows
    \draw[->, >=stealth, thick] (0, 2.7) -- (-1.2, 0.8) node[midway, above left=1pt] {$\frac{\gamma}{\lambda}dt$};
    \draw[->, >=stealth, thick] (0, 2.7) -- (1.2, 0.8) node[midway, above right=1pt] {$\frac{\alpha_i}{\lambda}dt$};

    % =========================================================
    % TREE 2: Start with \bar{0}
    % =========================================================
    % Time t
    \site{4.5}{3}
    \node[p2, anchor=south] at (4.5, 3.03) {$\bar{0}$};

    % Time t+dt
    \site{3.3}{0} % Empty 0

    \site{5.7}{0}
    \node[p1, anchor=south] at (5.7, 0.03) {$i$};

    % Arrows
    \draw[->, >=stealth, thick] (4.5, 2.7) -- (3.3, 0.8) node[midway, above left=1pt] {$\frac{\gamma}{\lambda}dt$};
    \draw[->, >=stealth, thick] (4.5, 2.7) -- (5.7, 0.8) node[midway, above right=1pt] {$\frac{\alpha_i}{\lambda}dt$};

    % =========================================================
    % TREE 3: Start with i
    % =========================================================
    % Time t
    \site{10}{3}
    \node[p1, anchor=south] at (10, 3.03) {$i$};

    % Time t+dt
    \site{8}{0} % Empty 0

    \site{10}{0}
    \node[p2, anchor=south] at (10, 0.03) {$\bar{0}$};

    \site{12}{0}
    \node[p1, anchor=south] at (12, 0.03) {$j$};

    % Arrows
    \draw[->, >=stealth, thick] (10, 2.7) -- (8, 0.8) node[midway, above left=1pt] {$\frac{\gamma}{\lambda}dt$};
    % Changed from midway to pos=0.65 to shift the label down
    \draw[->, >=stealth, thick] (10, 2.7) -- (10, 0.8) node[pos=0.65, right=2pt] {$\frac{\gamma}{\lambda}dt$};
    \draw[->, >=stealth, thick] (10, 2.7) -- (12, 0.8) node[midway, above right=1pt] {$\frac{\alpha_j}{\lambda}dt$};

    \end{tikzpicture}
    \caption{Graphical representation of the transition probabilities given by $B^{(1)}$ on the first site. The particle of species 0 is interpreted as an empty site.}
    \label{fig:transition_rates_Bp1}
\end{figure}

\begin{remark}
Note that the transition rates (\ref{eq:rates_B_(p,1)}) do not make a distinction between $0$ and $\bar 0$. This will be different for the second type of boundaries below.
\end{remark}

\begin{remark}\label{rem:uni_sect_(p,1)}
The boundaries $ B^{(1)} $ and $ \bar{B}^{(1)} $ allow the injection and extraction of any particle species at the boundaries. Thus it is easy to see that any configuration $ \vec\tau $ can be transformed for instance into the configuration filled with $0$'s (or into any other configuration $ \vec\tau'$). In other words, there is a unique sector (see Definition \ref{defi_sector} for the definition of sectors) for the $ (p,1) $-SSEP with boundaries $ B^{(1)} $ and $\bar B^{(1)}$. 
\end{remark}

\subsubsection{Definition of the boundaries $ B^{(2)} $}\label{sssec:def_(p,1)_B_(2)}

Now, we introduce a boundary that  makes a distinction between the species $0$ and $\bar 0$. Let
\begin{equation}\label{eq:k_(p,1)_B2}
k^{(2)} = \frac{2}{\lambda}\left(
\begin{array}{ccccc}
\gamma_{1} & \gamma_{1} & 0 & \dots & 0\\
\gamma_{2} & \gamma_{2} & 0 & \dots & 0 \\
0 & 0 & \alpha_1 & \dots & \alpha_1 \\
0 & 0 & \alpha_2 & \dots & \alpha_2 \\
\vdots & \vdots & \vdots &  & \vdots \\
0 & 0 & \alpha_p & \dots & \alpha_p
\end{array}
\right)\,\ \ \ \ \ \text{with $\gamma_1+\gamma_2=\sum_{i=1}^p\alpha_i=1$.}
\end{equation}
\begin{proposition}
The matrix $k^{(2)}$ satisfies the constant reflection equation with the matrix $\check r^{(p,1)}$:
\[\rc^{(p,1)}_{12}k^{(2)}_{1}\rc^{(p,1)}_{12}k^{(2)}_{1}=
k^{(2)}_{1}\rc^{(p,1)}_{12}k^{(2)}_{1}\rc^{(p,1)}_{12}\ .\]
\end{proposition}
\begin{proof}
This is a special case of Proposition \ref{prop_reflection_(p,q)} proved in Section \ref{sec:(p,q)-SSEP}.
\end{proof}
The matrix $k^{(2)}$ satisfies $(k^{(2)})^2=\frac{2}{\lambda}k^{(2)}$ and therefore, exactly as in the preceding subsection, we arrive at 
\begin{align}\label{eq:B_(p,1)_B2}
	B^{(2)} = \frac{1}{\lambda} \left(
\begin{array}{cccccc}
\gamma_1 & \gamma_1 & 0 & \dots & 0 \\
\gamma_2 & \gamma_2 & 0 & \dots & 0 \\
0 & 0 & \alpha_1 & \dots & \alpha_1 \\
0 & 0 & \alpha_2 & \dots & \alpha_2 \\
\vdots & \vdots & & & \vdots \\
0 & 0 & \alpha_p & \dots & \alpha_p 
\end{array}
\right)-\frac{1}{\lambda}\mathrm{Id}\,,
\end{align}
\begin{align}\label{eq:B_(p,1)_bar_B2}
\bar B^{(2)} = \frac{1}{\mu} \left(
\begin{array}{cccccc}
\beta_1 & \beta_1 & 0 & \dots & 0 \\
\beta_2 & \beta_2 & 0 & \dots & 0 \\
0 & 0 & \delta_1 & \dots & \delta_1 \\
0 & 0&\delta_2 & \dots & \delta_2 \\
\vdots & \vdots & & & \vdots \\
0 & 0 & \delta_p & \dots & \delta_p 
\end{array}
\right)-\frac{1}{\mu}\mathrm{Id}\,.
\end{align}
We choose all parameters to be positive and the constraints
\begin{equation}
\sum_{i=1}^{p} \alpha_{i}= \gamma_1+\gamma_2=1\,,\qquad \sum_{i=1}^{p} \delta_{i}=\beta_1+\beta_2=1 
\end{equation}
ensure that we get Markov matrices.

From $ B^{(2)} $ and $\bar B^{(2)}$, we read the transition rates at the boundaries. For example, we read from $ B^{(2)} $ the following transition rates at the left boundary:
\begin{equation}\label{eq:rates_B_(p,1)_B2}
	0 \xrightleftharpoons[\gamma_1/\lambda]{\gamma_2/\lambda}\bar 0,\quad  i \xrightleftharpoons[\alpha_{i}/\lambda]{\alpha_{j}/\lambda}j\,,
\end{equation}
while the transitions between $i$ and $0,\bar{0}$ for $ i\in \llbracket 1,p\rrbracket $ are forbidden.

\begin{figure}[H]
	\centering
\begin{tikzpicture}[
    scale=0.85, transform shape,
    % Particle i, j: Black fill, white text
    p1/.style={circle, draw=black, fill=black, text=white, minimum size=14pt, inner sep=0pt, font=\bfseries},
    % Particle \bar{0}: Gray fill, thick border, black text
    p2/.style={circle, draw=black, thick, fill=gray!20, minimum size=14pt, inner sep=0pt, font=\small\bfseries}
]

% Define a macro for the first site
% It draws the left wall, bottom floor, right wall, AND a horizontal continuation
\def\site#1#2{
    \draw (#1-0.4, #2+0.15) -- (#1-0.4, #2) -- (#1+0.4, #2) -- (#1+0.4, #2+0.15);
    \draw (#1+0.4, #2) -- (#1+0.8, #2); % Extension to the right
}

% =========================================================
% TIME LABELS
% =========================================================
\node[anchor=east, font=\bfseries] at (-2.5, 3.0) {Time $t$};
\node[anchor=east, font=\bfseries] at (-2.5, 0.0) {Time $t+dt$};

% =========================================================
% TREE 1: Start with 0 (Empty Site)
% =========================================================
% Time t
\site{0}{3}

% Time t+dt
\site{0}{0}
\node[p2, anchor=south] at (0, 0.03) {$\bar{0}$};

% Arrows
\draw[->, >=stealth, thick] (0, 2.7) -- (0, 0.8) node[midway, right=2pt] {$\frac{\gamma_2}{\lambda}dt$};

% =========================================================
% TREE 2: Start with \bar{0}
% =========================================================
% Time t
\site{5}{3}
\node[p2, anchor=south] at (5, 3.03) {$\bar{0}$};

% Time t+dt
\site{5}{0} % Empty 0

% Arrows
\draw[->, >=stealth, thick] (5, 2.7) -- (5, 0.8) node[midway, right=2pt] {$\frac{\gamma_1}{\lambda}dt$};

% =========================================================
% TREE 3: Start with i
% =========================================================
% Time t
\site{10}{3}
\node[p1, anchor=south] at (10, 3.03) {$i$};

% Time t+dt
\site{10}{0}
\node[p1, anchor=south] at (10, 0.03) {$j$};

% Arrows
\draw[->, >=stealth, thick] (10, 2.7) -- (10, 0.8) node[midway, right=2pt] {$\frac{\alpha_j}{\lambda}dt$};

\end{tikzpicture}
\caption{Graphical representation of the transition probabilities given by $B^{(2)}$ on the first site. The particle of species 0 is interpreted as an empty site.}
    \label{fig:transition_rates_Bp2}
\end{figure}

\begin{remark}\label{rem:sect_(p,1)_B2}
Contrary to the boundaries $B^{(1)}$ and $\bar B^{(1)}$ (see Remark \ref{rem:uni_sect_(p,1)}), the boundaries $ B^{(2)} $ and $ \bar{B}^{(2)} $ do not lead to a single sector. Indeed it is easy to see that the total number of particles of species $0$ and $\bar 0$ will remain constant all along the process, and that this number uniquely characterises a sector. Therefore there are $L+1$ different sectors, and the number of configurations in a sector is
\begin{equation}\label{eq:nb_sec_(p,1)_B2}
	\begin{pmatrix}
		L\\
		n_{0 \bar{0}}
	\end{pmatrix}
	2^{n_{0 \bar{0}}}p^{L-n_{0 \bar{0}}}\,,
\end{equation}
where $n_{0\bar0}$ is the fixed number of particles of species $0$ and $\bar 0$ in the sector.
\end{remark}

\begin{remark}\label{rem:B1B2}
One can also consider the Markov matrix 
$$
M = B^{(1)}_{1}+\sum_{i=1}^{L-1}m_{i,i+1}+\bar{B}^{(2)}_{L}\,,
$$ 
constructed from the transfer matrix
$$
t(u)=\mathrm{tr}_{0}\left( \tilde{k}^{(2)}_{0}(u)r_{0L}(u)\dots r_{01}(u)k^{(1)}_{0}(u)r_{10}(u)\dots r_{L0}(u)\right).
$$
In that case, the left boundary couples all species, while the right boundary distinguishes between the species $0$ and $\bar 0$. 
\end{remark}

\section{Generalisation to the $(p,q)$-SSEP model}\label{sec:(p,q)-SSEP}

In this section, we define the $ (p,q) $-SSEP which generalises the $ (p,1) $-SSEP. We first define this model in the periodic case and discuss some of its symmetries. For the open case, we generalise to the $(p,q)$-SSEP the boundaries $ B^{(1)} $ and $ B^{(2)} $ previously defined in the $ (p,1) $-SSEP case.

\subsection{The periodic $ (p,q) $-SSEP}\label{ssec:(p,q)-SSEP}
 
Let $ p,q \geq0 $. We choose the set $X$ to be of size $p+2q$ and we index its elements as follows:
\begin{equation}\label{eq:set_(p,q)}
	X=\{1,2,\dots,p,1^{-},1^{+},2^{-},2^{+},\dots,q^{-},q^{+}\}\,.
\end{equation}
The symbols $ a^{-},a^{+} $ for $ a\in \llbracket 1,q \rrbracket $ represent the pairs of reactive particle species. The states $ 1,\dots,p $ correspond to $ p $ singlet particle species.

Let $i,j\in\llbracket 1,p \rrbracket$ and $ a,b\in \llbracket 1,q \rrbracket $ with $ a\neq b $. We define the set-theoretical matrix $\rc$ by:
\begin{equation}\label{def:rc(p,q)}
\rc:\ 
\begin{array}{l}  \lvert a^{-}a^{-}\rangle \mapsto \lvert a^{+} a^{+}\rangle\,, \\
\lvert a^{+} a^{+}\rangle \mapsto \lvert a^{-} a^{-}\rangle\,, \\
\lvert a^{-}a^{+}\rangle \mapsto \lvert a^{-} a^{+}\rangle\,, \\
\lvert a^{+} a^{-}\rangle \mapsto \lvert a^{+} a^{-}\rangle\,,\end{array}\quad 
\begin{array}{l}  \lvert a^{\pm}\,i\rangle \mapsto \lvert i\, a^{\pm}\rangle\,, \\
\lvert i\, a^{\pm}\rangle \mapsto \lvert a^{\pm}\, i\,\rangle\,, 
\end{array}\quad 
\begin{array}{l}  \lvert a^{\pm}\,b^{\pm}\rangle \mapsto \lvert b^{\pm}\, a^{\pm}\rangle\,, \\
\lvert  a^{\pm}\,b^{\mp}\rangle \mapsto \lvert b^{\mp}\,a^{\pm}\rangle\,, 
\end{array}\quad 
\lvert i\,j\rangle \mapsto \lvert j\,i\rangle\ .
\end{equation}
In other words, on each subspace with basis $(\lvert a^{-},a^{-} \rangle, \lvert a^{-},a^{+} \rangle, \lvert a^{+},a^{-} \rangle, \lvert a^{+},a^{+} \rangle)$, the matrix $\rc$ acts as follows:
\begin{equation}\label{eq:R_(1)_def_(p,q)}
\begin{pmatrix}
\cdot & \cdot & \cdot & 1\\
\cdot & 1 & \cdot & \cdot\\
\cdot & \cdot & 1 & \cdot\\
1 & \cdot & \cdot & \cdot
\end{pmatrix}\,,
\end{equation}
while on the remaining basis vectors, the matrix $\rc$ acts simply as the permutation.

\begin{remark}\label{rema_q=1}
The previous definition of $\rc$ gives back the set-theoretical matrix of the $(p,1)$-SSEP given in Equation \eqref{def:rc(p,1)} for $q=1$. Note that the reactive pair $\{1^-,1^+\}$ corresponds in this case to the pair $\{0,\bar{0}\}$ and that we use a slightly different order here by putting the reactive species after the singlet particles.
\end{remark}

\begin{remark}\label{rem:multi_ssep}
If no reactive particle species are present, i.e. $ q=0 $, the $ (p,0) $-SSEP is the same as the multi-species SSEP with $ p $ particle species. 
\end{remark}

\begin{proposition}\label{prop:YBE_r_(p,q)}
The matrix $ \rc $ is solution of the YBE.
\end{proposition}
\begin{proof}
We need to prove that for all $x,x',x''\in X $, we have 
\begin{equation}\label{eq:yb_triplet}
	\rc_{12}\rc_{23}\rc_{12}|x,x',x''\rangle=\rc_{23}\rc_{12}\rc_{23}|x,x',x''\rangle\,. 
\end{equation}
First assume that the triplet $(x,x',x'')$ contains at most one element of any given reactive pair $ \{a^{-},a^{+}\}$. In this case, the matrix $ \rc $ acts as the permutation $P$ so it satisfies the YBE equation. 

The remaining cases can be checked by direct inspection. We find
\begin{align}
\rc_{12}\rc_{23}\rc_{12}|a^{-},a^{-},a^{-}\rangle=\rc_{23}\rc_{12}\rc_{23}|a^{-},a^{-},a^{-}\rangle=|a^{-},a^{-},a^{-}\rangle\,,\\
\rc_{12}\rc_{23}\rc_{12}|a^{-},a^{-},a^{+}\rangle=\rc_{23}\rc_{12}\rc_{23}|a^{-},a^{-},a^{+}\rangle=|a^{+},a^{-},a^{-}\rangle\,,\\
\rc_{12}\rc_{23}\rc_{12}|a^{-},a^{+},a^{-}\rangle=\rc_{23}\rc_{12}\rc_{23}|a^{-},a^{+},a^{-}\rangle=|a^{-},a^{+},a^{-}\rangle\,,
\end{align}
in the case where all elements of the triplet belong to the same set $\{a^{-},a^{+}\}$. In the case where exactly two elements are in $ \{ a^{-}, a^{+} \} $ and with $ j\notin \{ a^{-},a^{+} \} $, we have
\begin{align}	\rc_{12}\rc_{23}\rc_{12}|a^{-},a^{-},j\rangle=\rc_{23}\rc_{12}\rc_{23}|a^{-},a^{-},j\rangle=|j,a^{+},a^{+}\rangle\,,\\
\rc_{12}\rc_{23}\rc_{12}|a^{-},j,a^{-}\rangle=\rc_{23}\rc_{12}\rc_{23}|a^{-},j,a^{-}\rangle=|a^{+},j,a^{+}\rangle\,,\\
\rc_{12}\rc_{23}\rc_{12}|a^{-},a^{+},j\rangle=\rc_{23}\rc_{12}\rc_{23}|a^{-},a^{+},j\rangle=|j,a^{-},a^{+}\rangle\,,\\
\rc_{12}\rc_{23}\rc_{12}|a^{-},j,a^{+}\rangle=\rc_{23}\rc_{12}\rc_{23}|a^{-},j,a^{+}\rangle=|a^{-},j,a^{+}\rangle\,.
\end{align}    
From the involutivity of $\check r$ and the obvious symmetry exchanging $a^{-}$ and $a^{+}$, all cases are covered and this concludes the proof.
\end{proof}
\begin{remark}
In the usual notation for set-theoretical solutions of YBE, see Section \ref{ssec:set_theo}, the matrix $\check r$ is defined by
\begin{equation}\label{eq:r_(p,q)_set_1}
\rc \lvert i,j \rangle = \lvert f_{i}(j), f_{j}(i) \rangle\,,	
\end{equation}
where $ f_{a^{-}}=f_{a^{+}} $ is the transposition $ (a^{-},a^{+}) $ for $ a\in \llbracket 1,q \rrbracket $ and $ f_{i} $ is the identity for $ i\in \llbracket 1,p \rrbracket$. The matrix $ \rc $ is non-degenerate, involutive, symmetric and is not a Lyubashenko solution if $q\neq0 $ and $(p,q)\neq (0,1)$.  
\end{remark}
The Markov matrix $M$ for a periodic model is constructed from $\rc$ using Equations \eqref{eq:M_period} and \eqref{eq:loc_jump_r}. The allowed local configuration changes for $ i,j\in \llbracket 1,p \rrbracket $ and $ a, b\in \llbracket 1,q \rrbracket $ with $ a\neq b $ are
\begin{align}\label{eq:loc_move_(p,q)}
	&a^{-}a^{-} \longleftrightarrow a^{+}a^{+},\quad ij \longleftrightarrow ji,\quad a^{\pm}i \longleftrightarrow ia^{\pm}\,,\\
	&a^{\pm}b^{\pm}\longleftrightarrow b^{\pm}a^{\pm},
\quad a^{\pm}b^{\mp}\longleftrightarrow b^{\mp}a^{\pm}\,.
\end{align}
Using the same construction as in Section \ref{ssec:int_(p,1)} and relying on Proposition \ref{prop:YBE_r_(p,q)}, the periodic $ (p,q) $-SSEP is integrable.

The transformation interpretation given for the $ (p,1) $-SSEP generalises naturally for the $ (p,q) $-SSEP. In this interpretation, the local changes $ a^{-}a^{-} \leftrightarrow a^{+}a^{+} $ are interpreted as transformation processes. The singlet particles can freely move by exchanging their position with other singlet or reactive particles. The reactive particles can also freely move except that two conjugated reactive particles cannot overtake each other.
This interpretation is illustrated in Figure \ref{fig:annihi_process_(p,q)} in the introduction.

\begin{remark}
As in Remark \ref{rem:interp_(p,1)}, if we choose any of the singlet species to represent an empty site, reactive particles can move across empty sites.
Alternatively, if we choose any of the reactive species to represent an empty site, the corresponding conjugate species is interpreted as an evaporating/condensing species while the remaining reactive species are subjected to transformations.
\end{remark}

\begin{remark}\label{rem:reac_diff_AR}
The process in \cite{Alcaraz1993} where multiple particle species can diffuse, interchange and annihilate has a similar interpretation as the $ (p,q) $-SSEP. However, the annihilation process in \cite{Alcaraz1993} involves more than two conjugate reactive species and has different exchange rules from those of the $ (p,q) $-SSEP.
\end{remark}

\subsubsection{Symmetry of the periodic model}\label{ssec:sym}
For any $x,y\in X$, we introduce the elementary matrices $E_{xy}$ such that $E_{xy}\lvert z\rangle=\delta_{y,z}\lvert x \rangle$ for $z\in X$.
We also introduce,
\begin{equation}
    H_a = E_{a^-a^-} +E_{a^+a^+} \,,\quad S_a = E_{a^-a^+} +E_{a^+a^-}\,,\quad \forall a\in \llbracket 1,q \rrbracket\,,
\end{equation} which are all commuting one with each other and with all $E_{ij}$, $i,j\in \llbracket 1,p \rrbracket$.
Altogether they form the algebra  
\begin{align}
  \mathfrak{g}_{p,q}&:=\mathfrak{gl}_p\oplus\underbrace{ \mathfrak{gl}_1\oplus\cdots\oplus \mathfrak{gl}_1}_{2q}\,.   
\end{align}  
The algebra  $\mathfrak{g}_{p,q}$ is the set of matrices $G$ satisfying the equation
\begin{equation}
    \big[ G\otimes \mathrm{Id} + \mathrm{Id}\otimes G, \check r\big]=0\,.
\end{equation}
Therefore, the algebra $\mathfrak{g}_{p,q}$ is a symmetry algebra of $\check r$, and thus a symmetry algebra for the periodic $ (p,q) $-SSEP, in the sense
\begin{equation}
    \big[ G\otimes \mathrm{Id}^{\otimes L-1}+\mathrm{Id}\otimes G\otimes \mathrm{Id}^{\otimes L-2}+\cdots +\mathrm{Id}^{\otimes L-1}\otimes G\,,\, M]=0\,,\quad \forall\, G\in\mathfrak{g}_{p,q}\,.
\end{equation}
\begin{remark}
We observe that for $ q=0 $, i.e. the process reduces to the multi-species SSEP with $ p $ species, the symmetry algebra is $ \mathfrak{gl}_p $ so for $ p\geq 2 $ it contains $ \mathfrak{su}(2) $ subalgebras as in the diffusion-reaction process of \cite{Schutz1994}. In the presence of reactive species, the $ \mathfrak{su}(2) $ subalgebras remain for $ p\geq2 $ but their action is trivial on reactive species vectors. 
\end{remark}

\subsection{The $(p,q)$-SSEP with open boundaries}\label{ssec:bound_(p,q)}

In the following, we will generalise the integrable boundaries $ B^{(1)} $ and $ B^{(2)} $ of Section \ref{ssec:bound_(p,1)} for the $ (p,q) $-SSEP.

\subsubsection{A class of solution of the reflection equation}\label{sssec:sol_re_(p,q)}

Fix $q=q_1+q_2$ with $q_1,q_2\geq 0$. This corresponds to a splitting of the set of reactive species into two subsets
\[\{1^-,1^+,\dots,q_1^-,q_1^+\}\ \ \ \ \text{and}\ \ \ \ \{(q_1+1)^-,(q_1+1)^+,\dots,(q_1+q_2)^-,(q_1+q_2)^+\}\]
which are going to behave differently with respect to the boundaries.

Let $k$ be the following block-diagonal matrix
\begin{equation}\label{kblockdiagonal}
    k=\frac{2}{\lambda}\left(\begin{array}{c|c} k' & 0 \\ \hline 0 & k''
    \end{array}\right)\ ,
\end{equation}
with
\begin{equation}\label{k'k''}
    k'=\left(\begin{array}{cccc} \alpha_1  & \alpha_1 & \cdots\ \cdots\ \cdots\ \cdots & \alpha_1 \\
    \vdots & \vdots & & \vdots \\
    \vdots & \vdots & & \vdots \\
    \alpha_p & \alpha_p & \cdots\ \cdots\ \cdots\ \cdots & \alpha_p \\
    \gamma_1 & \gamma_1 &  \cdots\ \cdots\ \cdots\ \cdots & \gamma_1 \\
    \gamma_1 & \gamma_1 & \cdots\ \cdots\ \cdots\ \cdots & \gamma_1 \\
    \vdots & \vdots & & \vdots \\
    \gamma_{q_1} & \gamma_{q_1} & \cdots\ \cdots\ \cdots\ \cdots & \gamma_{q_1} \\
    \gamma_{q_1} & \gamma_{q_1} & \cdots\ \cdots\ \cdots\ \cdots & \gamma_{q_1} 
    \end{array}\right), \ \ \ k''=\left(\begin{array}{cccc} \!\!\!\begin{array}{cc} \gamma^-_1  & \gamma^-_1 \\[0.4em]
    \gamma^+_1  & \gamma^+_1 \end{array} & 0 & \cdots & 0 \\
     0 &  \ddots & & \vdots  \\
    \vdots & & \ddots & 0\\
     0 & \cdots & 0 & \begin{array}{cc} \gamma^-_{q_2}  & \gamma^-_{q_2} \\[0.4em]
    \gamma^+_{q_2}  & \gamma^+_{q_2} \end{array}\!\!\!
    \end{array}\right),
\end{equation}
and the constraints
\begin{equation}\label{constraints_(p,q)}
   \sum_{i=1}^p\alpha_i+2\sum_{a=1}^{q_1}\gamma_a=1\ \ \ \ \text{and}\ \ \ \ \gamma_b^-+\gamma_b^+=1\,,\ \forall b=1,\dots,q_2\,.
\end{equation}
The first block $k'$ is a square matrix of size $p+2q_1$. It acts on the subspace of $V$ with basis vectors $|1\rangle,\dots,|p\rangle,|1^-\rangle,|1^+\rangle,\dots, |q_1^-\rangle,|q_1^+\rangle$ as shown in (\ref{k'k''}). Note that $k'$ is of rank $1$ (all rows/columns are the same), and that the lines corresponding to a pair of reactive species $a^-,a^+$ are equal. The action of $k'$ is given explicitly by
\begin{equation}\label{actionk'}
k'|x\rangle=\sum_{i=1}^p\alpha_i|i\rangle+\sum_{a=1}^{q_1}\gamma_a(|a^-\rangle+|a^+\rangle)\,,\ \ \ \ \ \forall x\in\{1,\dots,p,1^-,1^+,\dots,q_1^-,q_1^+\}\ .
\end{equation}
The matrix $k''$ is a square matrix of size $2q_2$. It acts on the subspace of $V$ with basis vectors $|(q_1+1)^-\rangle,|(q_1+1)^+\rangle,\dots, |(q_1+q_2)^-\rangle,|(q_1+q_2)^+\rangle$ as shown in (\ref{k'k''}). Note that $k''$ itself has a block-diagonal form. Actually for any $a\in\llbracket q_1+1,q_1+q_2\rrbracket$ it acts on the subspace spanned by $|a^-\rangle,|a^+\rangle$. The action is given by
\begin{equation}\label{actionk''}
k''|a^{\pm}\rangle=\gamma_{a-q_1}^-|a^-\rangle+\gamma_{a-q_1}^+|a^+\rangle\,,\ \ \ \ \ \forall a\in\llbracket q_1+1,q_1+q_2\rrbracket\ .
\end{equation}
\begin{remark}
The matrix $k$ defined in \eqref{kblockdiagonal} generalises both matrices $k^{(1)}$ and $k^{(2)}$ of Section \ref{ssec:bound_(p,1)}, which are recovered, respectively, when $(q_1,q_2)=(1,0)$ and $(q_1,q_2)=(0,1)$. Note that the pair $\{0,\bar{0}\}$ corresponds now to $\{1^-,1^+\}$ and that the basis has been reordered slightly.
\end{remark}

\begin{proposition}\label{prop_reflection_(p,q)}
The matrix $k$ satisfies the constant reflection equation with the matrix $\check r$:
\begin{equation}\label{constant_ref_pq}\rc_{12}k_{1}\rc_{12}k_{1}=
k_{1}\rc_{12}k_{1}\rc_{12}\ .\end{equation}
\end{proposition}
\begin{proof}
We will check Equation (\ref{constant_ref_pq}) on all vectors $|x,y\rangle$ with $x,y\in X$. To do so efficiently, consider the following decomposition of $X$ into disjoint subsets:
\[X=X_0\sqcup X_{1}\sqcup\dots\sqcup X_{q_2}\ \ \ \ \text{with}\ \left\{\begin{array}{l} X_0=\{1,\dots,p,1^-,1^+,\dots,q_1^-,q_1^+\}\,,\\[0.4em]X_b=\{(q_1+b)^-,(q_1+b)^+\}\ \ \text{for $b=1,\dots,q_2$}\ .\end{array}\right.\]
Let $V_i$ be the subspace spanned by the basis vectors $|x\rangle$ with $x\in X_i$. The block-diagonal form of the matrix $k$ gives that $V_i$ is stable by $k$. 

Now assume that $x$ and $y$ belong to two different subsets $X_i$ and $X_j$. The definition of $\check r$ shows that in this case $\check r$ acts as the permutation, namely, we have $\check r|x,y\rangle=|y,x\rangle$. Therefore, we obtain immediately that
\[\rc_{12}k_{1}\rc_{12}k_{1}|x,y\rangle=
k_{1}\rc_{12}k_{1}\rc_{12}|x,y\rangle=k|x\rangle\otimes k|y\rangle\ .\
\]
Thus it remains to deal with the situation where $x,y$ belong to the same subset.

$\bullet$ Case 1. Consider the subset $X_b=\{(q_1+b)^-,(q_1+b)^+\}$ with $b\in\llbracket 1,q_2\rrbracket$. Define
\[|\Gamma_b\rangle=\frac{2}{\lambda}\left(\gamma_b^-|(q_1+b)^-\rangle+\gamma_b^+|(q_1+b)^+\rangle\right)\ \ \ \ \text{and}\ \ \ \ |\bar\Gamma_b\rangle=\frac{2}{\lambda}\left(\gamma_b^-|(q_1+b)^+\rangle+\gamma_b^+|(q_1+b)^-\rangle\right)\ .\]
From \eqref{kblockdiagonal} and (\ref{k'k''}), we read immediately
\[k|x\rangle=|\Gamma_b\rangle\,,\ \ \ \ \ \forall x\in X_b\ .\]
Moreover, straightforward calculations show that
\[\check r|\Gamma_b\rangle\otimes |(q_1+b)^{\pm}\rangle=|(q_1+b)^{\mp}\rangle\otimes |\bar\Gamma_b\rangle\,,\ \ \ \ \ \ \ \ \check r|\Gamma_b\rangle\otimes |\bar \Gamma_b\rangle=|\Gamma_b\rangle\otimes |\bar \Gamma_b\rangle\ .\]
Now let $x,y\in X_b$. Noting that $\check r|x,y\rangle=|x',y'\rangle$ for some $x',y'\in X_b$, we obtain that
\[\rc_{12}k_{1}\rc_{12}k_{1}|x,y\rangle=
k_{1}\rc_{12}k_{1}\rc_{12}|x,y\rangle=|\Gamma_b\rangle\otimes |\bar \Gamma_b\rangle\ .\]

$\bullet$ Case 2. Consider the subset $X_0=\{1,\dots,p,1^-,1^+,\dots,q_1^-,q_1^+\}$.
Define
\[|\Gamma_0\rangle=\frac{2}{\lambda}\left(\sum_{i=1}^p\alpha_i|i\rangle+\sum_{b=1}^{q_1}\gamma_b(|b^-\rangle+|b^+\rangle)\right)\ .\]
From (\ref{k'k''}), we read immediately
\[k|x\rangle=|\Gamma_0\rangle\,,\ \ \ \ \ \forall x\in X_0\ .\]
Moreover, straightforward calculations show that, for $i\in\llbracket 1,p\rrbracket$ and $b\in\llbracket 1,q_1\rrbracket$,
\[\check r|\Gamma_0\rangle\otimes |i\rangle=|i\rangle\otimes |\Gamma_0\rangle\,,\ \ \ \check r|\Gamma_0\rangle\otimes |b^{\pm}\rangle= | b^\pm\rangle\otimes | \Gamma_0\rangle \pm\frac{2\gamma_b}{\lambda}(| b^-\rangle- |b^+\rangle)\otimes (|b^-\rangle +|b^+\rangle)\ .\]
Recall that $k|b^-\rangle=k|b^+\rangle=|\Gamma_0\rangle$ and therefore $k_1\check r|\Gamma_0\rangle\otimes |b^{\pm}\rangle=|\Gamma_0\rangle\otimes | \Gamma_0\rangle$. It also follows from the above calculations that
\[ \check r|\Gamma_0\rangle\otimes | \Gamma_0\rangle=|\Gamma_0\rangle\otimes |\Gamma_0\rangle\ .\]
Now let $x,y\in X_0$. Noting that $\check r|x,y\rangle=|x',y'\rangle$ for some $x',y'\in X_0$, we obtain that
\[\rc_{12}k_{1}\rc_{12}k_{1}|x,y\rangle=
k_{1}\rc_{12}k_{1}\rc_{12}|x,y\rangle=|\Gamma_0\rangle\otimes |\Gamma_0\rangle\ .\]
This concludes the proof.
\end{proof}

\subsubsection{Baxterisation and boundary matrices}\label{sssec:baxt_bound_B}

The Baxterisation of the matrix $k$ follows exactly the same lines as the Baxterisation of $k^{(1)}$ and $k^{(2)}$ of Section \ref{ssec:bound_(p,1)}. We obtain
\begin{equation}\label{Bblockdiagonal_Bbarblockdiagonal}
    B=\frac{1}{\lambda}\left(\renewcommand{\arraystretch}{1.3}\begin{array}{c|c} B' & 0 \\ \hline 0 & B''
    \end{array}\right)\ ,\quad \bar{B}=\frac{1}{\mu}\left(\renewcommand{\arraystretch}{1.3}\begin{array}{c|c} \bar{B}' & 0 \\ \hline 0 & \bar{B}''
    \end{array}\right)\ ,
\end{equation}
with the same structure as in \eqref{kblockdiagonal} and (\ref{k'k''}), namely,
\begin{align}
    B'&=\scalebox{0.9}{$\displaystyle\left(\begin{array}{cccc} \alpha_1  & \alpha_1 & \cdots\cdots & \alpha_1  \\
    \vdots & \vdots & & \vdots \\
    \alpha_p & \alpha_p & \cdots\cdots & \alpha_p \\
    \gamma_1 & \gamma_1 &  \cdots\cdots & \gamma_1 \\
    \gamma_1 & \gamma_1 & \cdots\cdots & \gamma_1 \\
    \vdots & \vdots & & \vdots \\
    \gamma_{q_1} & \gamma_{q_1} & \cdots\cdots  & \gamma_{q_1} \\
    \gamma_{q_1} & \gamma_{q_1} & \cdots\cdots  & \gamma_{q_1} 
    \end{array}\right)$}-\mathrm{Id}, & B''&=\scalebox{0.9}{$\displaystyle\left(\begin{array}{cccc} \!\!\!\begin{array}{cc} \gamma^-_1  & \gamma^-_1 \\[0.4em]
    \gamma^+_1  & \gamma^+_1 \end{array} & 0 & \cdots & 0 \\
     0 &  \ddots & & \vdots  \\
    \vdots & & \ddots & 0\\
     0 & \cdots & 0 & \begin{array}{cc} \gamma^-_{q_2}  & \gamma^-_{q_2} \\[0.4em]
    \gamma^+_{q_2}  & \gamma^+_{q_2} \end{array}\!\!\!
    \end{array}\right)$}-\mathrm{Id}, \label{B'B''} \\
\intertext{and}
    \bar{B}'&=\scalebox{0.9}{$\displaystyle\left(\begin{array}{cccc} \delta_1  & \delta_1 & \cdots\cdots  & \delta_1 \\
    \vdots & \vdots & & \vdots \\
    \delta_p & \delta_p & \cdots\cdots  & \delta_p \\
    \beta_1 & \beta_1 &  \cdots\cdots  & \beta_1 \\
    \beta_1 & \beta_1 & \cdots\cdots  & \beta_1 \\
    \vdots & \vdots & & \vdots \\
    \beta_{q_1} & \beta_{q_1} & \cdots\cdots  & \beta_{q_1} \\
    \beta_{q_1} & \beta_{q_1} & \cdots\cdots  & \beta_{q_1} 
    \end{array}\right)$}-\mathrm{Id}, & \bar{B}''&=\scalebox{0.9}{$\displaystyle\left(\begin{array}{cccc} \!\!\!\begin{array}{cc} \beta^-_1  & \beta^-_1 \\[0.4em]
    \beta^+_1  & \beta^+_1 \end{array} & 0 & \cdots & 0 \\
     0 &  \ddots & & \vdots  \\
    \vdots & & \ddots & 0\\
     0 & \cdots & 0 & \begin{array}{cc} \beta^-_{q_2}  & \beta^-_{q_2} \\[0.4em]
    \beta^+_{q_2}  & \beta^+_{q_2} \end{array}\!\!\!
    \end{array}\right)$}-\mathrm{Id}. \label{Bb'Bb''}
\end{align}

All the parameters are positive and the constraints
\begin{equation}\label{eq:const_para_(p,q)}
    \begin{cases}\displaystyle\sum_{i=1}^{p} \alpha_{i}+2\sum_{a=1}^{q_1} \gamma_{a}=\gamma_b^-+\gamma_b^+=1\\[2.1ex] 
    \displaystyle\sum_{i=1}^{p} \delta_{i}+2\sum_{a=1}^{q_1} \beta_{a}=\beta_b^-+\beta_b^+=1\end{cases}\text{ for }b=1,\dots,q_2
\end{equation} ensure that we get Markov matrices.

From $ B' $ for instance, we read the following transition rates at the left boundary:
\begin{align}
    &i \xrightleftharpoons[\alpha_{i}/\lambda]{\alpha_{j}/\lambda} j,\quad
    i \xrightleftharpoons[\alpha_{i}/\lambda]{\gamma_{b_1}/\lambda}b_1^{\pm},\quad
    b_1^{-} \overset{\gamma_{b_1}/\lambda}{\longleftrightarrow}b_1^{+}, \notag \\
    &b_1^{-}\xrightleftharpoons[\gamma_{b_1}/\lambda]{\gamma_{b_2}/\lambda} b_2^{\pm},\quad
    b_1^{+} \xrightleftharpoons[\gamma_{b_1}/\lambda]{\gamma_{b_2}/\lambda}b_2^{\pm}\,,\quad \ \ \forall b_1,b_2\in\llbracket 1,q_1\rrbracket,\ b_1\neq b_2,\  \forall i,j\in\llbracket 1,p\rrbracket,\ i\neq j \,. \label{eq:rates_B'}
\end{align}
Similarly, we read from $ B'' $:
\begin{align}\label{eq:rates_B''}
    b_{1}^- \xrightleftharpoons[\gamma_{b_1-q_{1}}^-/\lambda]{\gamma_{b_1-q_{1}}^+/\lambda}b_1^+\,,\quad\forall b_{1}\in \llbracket q_{1}+1,q_{1}+q_{2} \rrbracket\,,
\end{align} 
while transitions between $b_{1}^{\epsilon}$ and $b_{2}^{\epsilon'}$ are forbidden for $b_1,b_2\in\llbracket q_1+1,q_1+q_2\rrbracket$, with $\ b_1\neq b_2$ and $\epsilon,\epsilon'\in\{+,-\}$.

Since $B$ is block diagonal, the transitions between $x\in \{1,\dots,p,1^-,1^+,\dots,q_1^-,q_1^+\}$ and $y\in \{(q_1+1)^-,(q_1+1)^+,\dots,(q_1+q_2)^-,(q_1+q_2)^+\}$  are forbidden.

\begin{figure}[H]
    \centering
    \begin{tikzpicture}[
        scale=0.8, transform shape,
        % Singlets: Black fill, white text, 14pt
        singlet/.style={circle, draw=black, fill=black, text=white, minimum size=14pt, inner sep=0pt, font=\small\bfseries},
        % Reactive: Gray fill, thick border, black text, 14pt, font matches singlet
        reactive/.style={circle, draw=black, thick, fill=gray!20, text=black, minimum size=14pt, inner sep=0pt, font=\small\bfseries},
        % Superscript labels placed outside the contour at a 45-degree angle
        plus/.style={label={[font=\scriptsize\bfseries, inner sep=1pt, label distance=-2pt]45:$+$}},
        minus/.style={label={[font=\scriptsize\bfseries, inner sep=1pt, label distance=-2pt]45:$-$}},
        pm/.style={label={[font=\scriptsize\bfseries, inner sep=1pt, label distance=-2pt]45:$\pm$}},
        mp/.style={label={[font=\scriptsize\bfseries, inner sep=1pt, label distance=-2pt]45:$\mp$}}
    ]

    % Define a macro for the first site
    % It draws the left wall, bottom floor, right wall, AND a horizontal continuation
    \def\site#1#2{
        \draw (#1-0.4, #2+0.15) -- (#1-0.4, #2) -- (#1+0.4, #2) -- (#1+0.4, #2+0.15);
        \draw (#1+0.4, #2) -- (#1+0.8, #2); % Extension to the right
    }

    % =========================================================
    % TIME LABELS
    % =========================================================
    \node[anchor=east, font=\bfseries] at (-2.5, 3.0) {Time $t$};
    \node[anchor=east, font=\bfseries] at (-2.5, 0.0) {Time $t+dt$};

    % =========================================================
    % TREE 1: Start with b_1^{-}
    % =========================================================
    % Time t
    \site{0}{3}
    \node[reactive, minus, anchor=south] at (0, 3.03) {$b_1$};

    % Time t+dt
    \site{-1.8}{0}
    \node[singlet, anchor=south] at (-1.8, 0.03) {$i$};

    \site{0}{0}
    \node[reactive, plus, anchor=south] at (0, 0.03) {$b_1$};

    \site{1.8}{0}
    \node[reactive, pm, anchor=south] at (1.8, 0.03) {$b_2$};

    % Arrows
    \draw[->, >=stealth, thick] (0, 2.7) -- (-1.8, 0.8) node[midway, above left=1pt] {$\frac{\alpha_i}{\lambda}dt$};
    \draw[->, >=stealth, thick] (0, 2.7) -- (0, 0.8) node[pos=0.8, right=1pt] {$\frac{\gamma_{b_1}}{\lambda}dt$};
    \draw[->, >=stealth, thick] (0, 2.7) -- (1.8, 0.8) node[midway, above right=1pt] {$\frac{\gamma_{b_2}}{\lambda}dt$};

    % =========================================================
    % TREE 2: Start with b_1^{+}
    % =========================================================
    % Time t
    \site{5}{3}
    \node[reactive, plus, anchor=south] at (5, 3.03) {$b_1$};

    % Time t+dt
    \site{3.2}{0} 
    \node[singlet, anchor=south] at (3.2, 0.03) {$i$};

    \site{5}{0}
    \node[reactive, minus, anchor=south] at (5, 0.03) {$b_1$};

    \site{6.8}{0}
    \node[reactive, pm, anchor=south] at (6.8, 0.03) {$b_2$};

    % Arrows
    \draw[->, >=stealth, thick] (5, 2.7) -- (3.2, 0.8) node[midway, above left=1pt] {$\frac{\alpha_i}{\lambda}dt$};
    \draw[->, >=stealth, thick] (5, 2.7) -- (5, 0.8) node[pos=0.8, right=1pt] {$\frac{\gamma_{b_1}}{\lambda}dt$};
    \draw[->, >=stealth, thick] (5, 2.7) -- (6.8, 0.8) node[midway, above right=1pt] {$\frac{\gamma_{b_2}}{\lambda}dt$};

    % =========================================================
    % TREE 3: Start with i
    % =========================================================
    % Time t
    \site{10}{3}
    \node[singlet, anchor=south] at (10, 3.03) {$i$};

    % Time t+dt
    \site{8.5}{0} 
    \node[singlet, anchor=south] at (8.5, 0.03) {$j$};

    \site{11.5}{0}
    \node[reactive, pm, anchor=south] at (11.5, 0.03) {$b_1$};

    % Arrows
    \draw[->, >=stealth, thick] (10, 2.7) -- (8.5, 0.8) node[midway, above left=1pt] {$\frac{\alpha_j}{\lambda}dt$};
    \draw[->, >=stealth, thick] (10, 2.7) -- (11.5, 0.8) node[midway, above right=1pt] {$\frac{\gamma_{b_1}}{\lambda}dt$};

    \end{tikzpicture}
\caption{Graphical representation of the transition probabilities given by $\frac{1}{\lambda}B'$ on the first site with $i,j\in \llbracket1,p\rrbracket$, $i\neq j$, $b_1,b_2\in\llbracket1,q_1\rrbracket$, $b_1\neq b_2$.}
\label{fig:transition_rates_B'}
\end{figure}

\begin{figure}[H]
    \centering
    \begin{tikzpicture}[
        scale=0.85, transform shape,
        % Singlets: Black fill, white text, 14pt (kept for consistency with your preamble)
        singlet/.style={circle, draw=black, fill=black, text=white, minimum size=14pt, inner sep=0pt, font=\small\bfseries},
        % Reactive: Gray fill, thick border, black text, 14pt, font matches singlet
        reactive/.style={circle, draw=black, thick, fill=gray!20, text=black, minimum size=14pt, inner sep=0pt, font=\small\bfseries},
        % Superscript labels placed outside the contour at a 45-degree angle
        plus/.style={label={[font=\scriptsize\bfseries, inner sep=1pt, label distance=-2pt]45:$+$}},
        minus/.style={label={[font=\scriptsize\bfseries, inner sep=1pt, label distance=-2pt]45:$-$}}
    ]

    % Define a macro for the site
    % It draws the left wall, bottom floor, right wall, AND a horizontal continuation
    \def\site#1#2{
        \draw (#1-0.4, #2+0.15) -- (#1-0.4, #2) -- (#1+0.4, #2) -- (#1+0.4, #2+0.15);
        \draw (#1+0.4, #2) -- (#1+0.8, #2); % Extension to the right
    }

    % =========================================================
    % TIME LABELS
    % =========================================================
    \node[anchor=east, font=\bfseries] at (-2.5, 3.0) {Time $t$};
    \node[anchor=east, font=\bfseries] at (-2.5, 0.0) {Time $t+dt$};

    % =========================================================
    % TREE 1: a^{-} transitions to a^{+}
    % =========================================================
    % Time t
    \site{0}{3}
    \node[reactive, minus, anchor=south] at (0, 3.03) {$b_1$};

    % Time t+dt
    \site{0}{0}
    \node[reactive, plus, anchor=south] at (0, 0.03) {$b_1$};

    % Arrow
    \draw[->, >=stealth, thick] (0, 2.7) -- (0, 0.8) node[midway, right=2pt] {$\frac{\gamma_{b_1-q_1}^{+}}{\lambda}dt$};

    % =========================================================
    % TREE 2: a^{+} transitions to a^{-}
    % =========================================================
    % Time t
    \site{5}{3}
    \node[reactive, plus, anchor=south] at (5, 3.03) {$b_1$};

    % Time t+dt
    \site{5}{0}
    \node[reactive, minus, anchor=south] at (5, 0.03) {$b_1$};

    % Arrow
    \draw[->, >=stealth, thick] (5, 2.7) -- (5, 0.8) node[midway, right=2pt] {$\frac{\gamma_{b_1-q_1}^{-}}{\lambda}dt$};

    \end{tikzpicture}
    \caption{Graphical representation of the transition probabilities given by $\frac{1}{\lambda}B''$ on the first site with $b_1\in\llbracket q_1+1,q_1+q_2\rrbracket$.}
    \label{fig:transition_rates_B''}
\end{figure}

\begin{remark}\label{rem:sectors_open}
From the expression of $ B $, we observe that the total number of particles with species in $X_0=\{1,\dots,p,1^-,1^+,\dots,q_1^-,q_1^+\}$ is conserved by the dynamics. We denote by $ n_{X_0} $ this invariant. The total number of particle with species in $ X_b=\{ (q_1+b)^-,(q_1+b)^+ \} $ for each $b\in \llbracket 1,q_2 \rrbracket$ is also conserved and is denoted by $ n_{X_b} $. 
We observe that a sector is uniquely labeled by its profile $(n_{X_0},n_{X_1},\dots,n_{X_{q_2}})$. We deduce that the number of sectors is
\[
\begin{pmatrix}
	L+q_{2}\\
	L
\end{pmatrix}\,.\]
The size of a sector with profile $ (n_{X_0},n_{X_1},\dots,n_{X_{q_{2}}}) $ is
\[
\begin{pmatrix}
	L\\
n_{X_{0}},n_{X_{1}},\dots,n_{X_{q_{2}}}
\end{pmatrix}(p+2q_{1})^{n_{X_0}}2^{L-n_{X_0}}\,.\]	
\end{remark}

\begin{remark}
As for the $(p,1)$-SSEP, see remark \ref{rem:B1B2}, one can use two boundary matrices of different types. For instance, one may select for $\bar B$ a set of reactive pairs coupling to the singlet particles that is different from the one used in $B$. 
\end{remark}

\section{Physical quantities}\label{sec:phys_quant}

\subsection{Definitions}

In this part, we define the relevant physical quantities in the $ (p,q) $-SSEP. We first recall the definitions of the mean densities, lattice currents. Then, we introduce the mean reaction rate.

We suppose in the following that the periodic or open $ (p,q) $-SSEP has reached a stationary distribution $ S $ associated with a given sector of the process. The species evolving on the lattice of size $ L $ belong to a set $X$.

\begin{definition}
The mean density of a particle species $ s\in X $ at site $ i\in \llbracket 1,L \rrbracket$ in the stationary state is
\begin{equation}\label{eq:def_dens}
	\langle s \rangle_{i}= \sum_{\tau_{1},...,\tau_{i-1},\tau_{i+1},...,\tau_{L}\in X} S(\tau_{1},...,\tau_{i-1},s,\tau_{i+1},...,\tau_{L})\,.
\end{equation}
\end{definition}

We introduce the following quantity defined for species $s,s'\in X$ and sites $i\in \llbracket 1,L-1 \rrbracket 
$ \[\langle ss'\rangle_{i,i+1}=\sum_{\tau_{1},...,\tau_{i-1},\tau_{i+2},...,\tau_{L}\in X} S(\tau_{1},...,\tau_{i-1},s,s',\tau_{i+2},...,\tau_{L})\,,\]
which corresponds to the probability that a configuration has particle species $s$ at site $i$ and $s'$ at site $i+1$ in the stationary state.
We also recall that $  m(ss'\rightarrow s's)$ is the transition rate from the two site configuration $ ss' $ to $ s's $.
\begin{definition}
The mean lattice current of a particle species $ s\in X $ between site $ i $ and $ i+1 $ with $i\in \llbracket 1,L-1 \rrbracket$ in the stationary state is
\begin{equation}\label{eq:def_curr_alg}
	\langle j_{s}\rangle_{i,i+1}=\langle j_{s} \rangle_{i \rightarrow i+1}- \langle j_{s} \rangle_{i \leftarrow i+1}\,,
\end{equation}
with
\begin{equation}\label{eq:def_j}
	\langle j_{s} \rangle_{i \rightarrow i+1}=\sum_{\substack{s'\in X\\ s'\neq s}}\langle ss'\rangle_{i,i+1}  m(ss'\rightarrow s's),\quad \langle j_{s} \rangle_{i \leftarrow i+1}=\sum_{\substack{s'\in X\\ s'\neq s}}\langle s's\rangle_{i,i+1}m(s's\rightarrow ss')\,.
\end{equation}
\end{definition}
Recall that for us the transition rates $m(ss'\rightarrow s's)$ are always $1$ unless $s,s'$ is a pair of reactive species $a^+,a^-$ for which we have $m(a^+a^-\rightarrow a^-a^+)=m(a^-a^+\rightarrow a^+a^-)=0$.

The presence of reactive particles in the $ (p,q) $-SSEP allows us to introduce the following definition of reaction rate.
\begin{definition}\label{def:reac_rate}
The mean reaction rate of the reactive species $a^-$ with $a\in\llbracket 1,q\rrbracket$ at site $ i $ with $ i\in \llbracket 2,L-1 \rrbracket $ in the stationary state is
\begin{equation}\label{eq:reac_rate}
	\langle \nu_{a^-}\rangle_{i}=\langle \nu_{a^-} \rangle_{i-1,i}+\langle \nu_{a^-} \rangle_{i,i+1}\,,
\end{equation}
with
\begin{equation}\label{eq:reac_rate_sep}
	\langle \nu_{a^-} \rangle_{i,i+1}=\langle a^+a^+\rangle_{i,i+1}\,  m(a^+a^+\rightarrow a^-a^-)-\langle a^-a^-\rangle_{i,i+1}\,m(a^-a^-\rightarrow a^+a^+)\,.
\end{equation}
By definition, the mean reaction rate of $a^+$ is the opposite: $\langle \nu_{a^+}\rangle_{i}:=-\langle \nu_{a^-}\rangle_{i}$.
\end{definition}

Recall that in our model we have $m(a^+a^+\rightarrow a^-a^-)=m(a^-a^-\rightarrow a^+a^+)=1$.

We give the following two interpretations for the mean reaction rate. 
When the reactive particle species $a$ is interpreted as an evaporating/condensing species, which is naturally the case for $\bar{0}$ in the $(p,1)$-SSEP, the reaction rate of $a$ is interpreted as its mean condensation rate at site $i$.
When the reactive particle species $a$ is interpreted as a transforming species, the reaction rate of $a$ is interpreted as its mean production rate at site $i$.

\begin{remark}\label{rem:bound_curr}
In Definition \ref{eq:def_curr_alg}, the mean lattice current is defined in the bulk of the lattice. We also define the mean currents between the reservoirs and the lattice for a species $ s\in X $ as follows
\begin{align}\label{eq:left_right_curr}
	\langle j_{s} \rangle_{0,1}&=\sum_{\substack{s'\in X\\s'\neq s }}\left[ \langle s' \rangle_{1}B(s' \rightarrow s) - \langle s \rangle_{1}B(s \rightarrow s')\right]\,,\\
	\langle j_{s} \rangle_{L,L+1}&=\sum_{\substack{s'\in X\\s'\neq s }}\left[ \langle s \rangle_{L}\bar{B}(s \rightarrow s')-\langle s' \rangle_{L}\bar{B}(s' \rightarrow s)\right]\,.
\end{align}
where $ B(s \rightarrow s')= \langle s' \rvert B \lvert s \rangle $ and $ \bar{B}(s \rightarrow s')= \langle s' \rvert \bar{B} \lvert s \rangle $. The subscripts $0$ and $L+1$ stand for the position of the left and right reservoirs respectively.
\end{remark}

\subsection{Physical quantities in the case $q=q_{1}$}\label{ssec:case_q=q1}

In this section, the physical quantities defined in the previous section are obtained in the $(p,q)$-SSEP with the boundaries $B^{(1)}$, $\bar{B}^{(1)}$ corresponding to $B $, $\bar{B}$ respectively for $ q=q_{1} $. When there is no ambiguity, the right boundary $\bar{B}^{(1)}$ will not be mentioned in the following. We recall (see Section \ref{sssec:sol_re_(p,q)}) that when $q=q_1+q_2$ with $q_1,q_2\geq0$  the set of reactive species is split into the two subsets \[\{1^-,1^+,\dots,q_1^-,q_1^+\}\ \ \ \ \text{and}\ \ \ \ \{(q_1+1)^-,(q_1+1)^+,\dots,(q_1+q_2)^-,(q_1+q_2)^+\}.\]
In the case $q=q_1$, the first subset contains all the reactive species and the second one is empty.

The main result of this section is the following.
\begin{proposition}\label{prop:physical_quant_q=q1}
Let $ s,s'\in\llbracket 1,p\rrbracket$, $a,b\in\llbracket 1,q\rrbracket $. In the $ (p,q) $-SSEP with boundary $ B^{(1)} $, the mean densities are, for $i\in\llbracket 1,L\rrbracket$,
\begin{equation}\label{eq:dens_(p,q)_(p+q)}
\langle s \rangle_{i}=\frac{\alpha_{s}(L-i+\mu)+\delta_{s}(\lambda+i-1)}{L-1+\lambda+\mu}\,,\quad\langle a^\pm \rangle_{i}=\frac{\gamma_{a}(L-i+\mu)+\beta_{a}(\lambda+i-1)}{L-1+\lambda+\mu}\,,\end{equation} 
the mean lattice currents are, for $i\in\llbracket 1,L-1\rrbracket$,
\begin{equation}\label{eq:curr_(p,q)_(p+q)}
\langle j_{s}\rangle_{i,i+1}=\frac{\alpha_{s}-\delta_{s}}{L-1+\lambda+\mu},\qquad\langle j_{a^\pm}\rangle_{i,i+1}=\frac{\gamma_{a}-\beta_{a}}{L-1+\lambda+\mu}\,.
\end{equation}
The mean reaction rate $ \langle \nu_{a^\pm} \rangle_{i} $ of all reactive species vanishes for $ i\in \llbracket 2,L-1 \rrbracket $.
\end{proposition}

\begin{remark}\label{rem:eq_curr_dens}
The case where for all $s\in \llbracket 1,p\rrbracket,\, \alpha_{s}=\delta_{s} $ and for all $a\in \llbracket 1,q\rrbracket,\, \gamma_{a}=\beta_{a} $ corresponds to thermodynamic equilibrium. In that case, all mean currents vanish. Also, the mean densities are constant along the lattice (but they can be different for different species).
\end{remark}

\begin{remark}\label{rem:bound_curr_2}
Note that the mean currents $\langle j_{s}\rangle_{i,i+1}$ do not depend on the site $i$, that is, the mean current is constant along the lattice. Moreover, in Remark \ref{rem:bound_curr} we define mean currents at the boundaries. Using the Markovianity of $B^{(1)}$ and the equality $\sum_{s\in X} \langle s \rangle_1=1$ we find
\begin{align}\label{eq:left_curr_(p,q)}
	&\forall s\in \llbracket 1,p \rrbracket,\quad \langle j_{s} \rangle_{0,1}= \frac{1}{\lambda}\left(\alpha_{s}-\langle s \rangle_{1}\right)=\frac{\alpha_s-\delta_s}{L-1+\lambda+\mu}\,, \\
	&\forall a\in \llbracket 1,q \rrbracket,\quad \langle j_{a^{\pm}} \rangle_{0,1}=\frac{1}{\lambda}\left(\gamma_{a}-\langle a^{\pm}\rangle_{1}\right)=\frac{\gamma_a-\beta_a}{L-1+\lambda+\mu}\,.
\end{align}
For the right boundary, we similarly find $\langle j_{s} \rangle_{L,L+1}=\langle j_{s} \rangle_{0,1}$ and $\langle j_{a^{\pm}} \rangle_{L,L+1}=\langle j_{a^{\pm}} \rangle_{0,1}$. In particular, we see that the mean currents between the lattice and the reservoirs are equal to the mean current in the bulk.

Since the current satisfies in this case $\langle j_s \rangle_{i,i+1}=\langle s \rangle_{i}-\langle s \rangle_{i+1} $ for $i\in \llbracket1,L-1\rrbracket$, by extending this formula for $i=0$ and $i=L$, we define $\langle s \rangle_{0}$ and $\langle s \rangle_{L+1}$ which are interpreted as densities in the reservoirs. Then a straightforward calculation shows that, with these definitions, Formula (\ref{eq:dens_(p,q)_(p+q)}) is satisfied also for $i=0$ and $i=L+1$.
\end{remark}

\begin{remark}\label{rem:bound_phase_trans}
In the case of the thermodynamic limit, i.e. for $ L \rightarrow\infty $, the mean density profiles in Equation \eqref{eq:dens_(p,q)_(p+q)} become linear. The mean lattice currents in Equation \eqref{eq:curr_(p,q)_(p+q)} vanish for any boundary parameters in the thermodynamic limit. Since the mean reaction rates also vanish, there are no boundary-driven phase transitions. This is characteristic of purely diffusive systems \cite{Derrida2007}.
\end{remark}

In the remainder of this section, we prove Proposition \ref{prop:physical_quant_q=q1}. To do so, we determine the stationary state of the $ (p,q) $-SSEP with boundary $ B^{(1)} $ by connecting it to an open multi-species SSEP.

\subsubsection{Stationary distribution}\label{ssec:statio_q1}

We recall that the boundaries $ B $ and $ \bar{B} $ defined in Section \ref{sssec:baxt_bound_B} take the form for $ q=q_{1} $
\begin{equation}\label{eq:BBb_q=q1}
	B=\frac{1}{\lambda}B',\quad \bar{B}=\frac{1}{\mu}\bar{B}'\,,
\end{equation}
with $ B'$ and $ \bar{B}' $ defined in Equations \eqref{B'B''}, \eqref{Bb'Bb''}.

We now define the open multi-species SSEP onto which the $ (p,q) $-SSEP will be projected. The set of particle species of this process is
\begin{equation}\label{eq:X_msep}
	X^{(\mathrm{msep})}=\llbracket 1,p+q \rrbracket\,,
\end{equation}
and the corresponding vector space is $ V^{(\mathrm{msep})} $. 

We recall (see \cite{vanicat2017_2}) that in the multi-species SSEP, the local jump operator is $ m^{(\mathrm{msep})}=P-\mathrm{Id} $ with $ P $ the permutation operator. Several families of integrable boundaries exist but the one we are interested in are given in the basis $ ( \lvert 1 \rangle,\dots,\lvert p+q \rangle) $ of $ V^{\mathrm{(msep)}} $ by
\begin{equation}\label{eq:B_(p+q)_SSEP}
B^{(\mathrm{msep})} = \frac{1}{\lambda} \begin{pmatrix}
\alpha_{1} & \alpha_{1} & \dots & \alpha_{1} \\
\vdots & \vdots &  & \vdots \\
\alpha_{p} & \alpha_{p}  & \dots & \alpha_{p} \\
2\gamma_{1} & 2\gamma_{1}  & \dots & 2\gamma_{1} \\
\vdots & \vdots &  & \vdots \\
2\gamma_{q} & 2\gamma_{q}  & \dots & 2\gamma_{q} \\
\end{pmatrix}-\frac{1}{\lambda}\mathrm{Id}\,,
\end{equation}
\begin{equation}\label{eq:Bb_(p+q)_SSEP}
\bar{B}^{(\mathrm{msep})} = \frac{1}{\mu} \begin{pmatrix}
\delta_{1} & \delta_{1} & \dots & \delta_{1} \\
\vdots & \vdots &  & \vdots \\
\delta_{p} & \delta_{p} & \dots & \delta_{p} \\
2\beta_{1} & 2\beta_{1} & \dots & 2\beta_{1} \\
\vdots & \vdots &  & \vdots \\
2\beta_{q} & 2\beta_{q} & \dots & 2\beta_{q} \\
\end{pmatrix}-\frac{1}{\mu}\mathrm{Id}\,,
\end{equation}
where the parameters $\alpha_1,\dots,\alpha_p,\gamma_1,\dots,\gamma_q,\delta_1,\dots,\delta_p,\beta_1,\dots,\beta_q$ are the ones of $ B' $ and $ \bar{B}' $ of Equations \eqref{B'B''} and \eqref{Bb'Bb''} and $\lambda,\mu$ are the same as in Equation \eqref{eq:BBb_q=q1}.

The constraints
\begin{equation}\label{eq:cons_BBbmsep}
	\sum_{i=1}^{p} \alpha_{i}+2 \sum_{a=1}^{q} \gamma_{a}=\sum_{i=1}^{p} \delta_{i}+2 \sum_{a=1}^{q} \beta_{a}=1\,,
\end{equation}
ensure that $ B^{(\mathrm{msep})} $ and $ \bar{B}^{(\mathrm{msep})} $ are Markov matrices. Moreover, it is clear that those boundary matrices define an irreducible process, i.e. a process with a unique sector.

The projection of the $ (p,q) $-SSEP on this open multi-species SSEP is done by identifying the two species within each pair of reactive species. Denoting by $ \mathfrak{C} $ and $ \mathfrak{C}^{(\mathrm{msep})} $ the sets of configurations of the respective processes, the identification is performed by the map
\begin{equation}\label{eq:Phi_surj_(p,q)}
\Phi :
\mathfrak{C}
\longrightarrow
\mathfrak{C}^{(\mathrm{msep})},
\qquad
\Phi = \prod_{i=1}^{L} \phi_i \, .
\end{equation}
with
\begin{equation}\label{eq:phi_surj_(p,q)}
\begin{array}{c}
\phi : X \longrightarrow X^{\mathrm{(msep)}} \\[1ex]
\begin{array}{rcl @{\qquad\qquad} rcl}
1 & \longmapsto & 1      &  1^{\pm}     & \longmapsto & p+1 \\
2 & \longmapsto & 2      &  2^{\pm}     & \longmapsto & p+2 \\
\vdots    &             & \vdots & \vdots &             & \vdots \\
p & \longmapsto & p    &  q^{\pm}    & \longmapsto & p+q
\end{array}
\end{array}\,.
\end{equation}
We denote
$$n_{\mathrm{reac}}(\vec\tau)=|\{i\in\,\llbracket 1,L\rrbracket\ \text{with}\ \tau_i\in\{1^{\pm},\dots,q^{\pm}\}\}|$$ 
the total number of reactive particles in a configuration $\vec\tau$. Note that
\[2^{n_{\mathrm{reac}}(\vec\tau)}=|\{\vec\tau'\ \text{such that}\ \Phi(\vec\tau')=\Phi(\vec\tau)\}|\ ,
\]
that is, the number of configurations having the same projection $\Phi$ as $\vec\tau$ is given by $2^{n_{\mathrm{reac}}(\vec\tau)}$.
\begin{proposition}\label{prop:statio_(p,q)_(p+q)}
The stationary distribution $ S $ of the $ (p,q) $-SSEP with boundaries $ B^{(1)} $ and $\bar B^{(1)} $ is related to the stationary distribution $ S^{(\mathrm{msep})} $ of the multi-species SSEP with boundaries $ B^{(\mathrm{msep})} $ and $ \bar B^{(\mathrm{msep})} $ as
\begin{equation}\label{eq:statio_(p,q)_(p+q)}
	\forall \vec\tau\in \mathfrak{C}\,,\ \ 	S(\vec\tau)=\frac{1}{2^{n_{\mathrm{reac}}(\vec\tau)}}S^{(\mathrm{msep})}(\Phi(\vec\tau))\,.
    \end{equation}
\end{proposition}

\begin{proof}
Define the map $ u : V^{(\mathrm{msep})} \rightarrow V $ as
\begin{equation}\label{eq:def_u_(p,q)}
\forall i\in \llbracket 1,p \rrbracket\,,\ u \lvert i \rangle = \lvert i\rangle,\quad\forall i\in \llbracket p+1,p+q \rrbracket\,,\ u \lvert i \rangle = \frac{1}{2}( \lvert (i-p)^{-} \rangle+\lvert (i-p)^{+} \rangle)\,.
\end{equation}
It can be easily verified that the operator $ u $ satisfies
\begin{equation}\label{eq:uinv_(p,q)_local_rel}
	u B^{(\mathrm{msep})}= B u,\quad u\otimes u \ m^{(\mathrm{msep})} = m\ u \otimes u,\quad u \bar{B}^{(\mathrm{msep})}=\bar{B} u\,.
\end{equation}
With the non-invertible map $ U : (V^{(\mathrm{msep})})^{ \otimes L} \rightarrow V^{\otimes L} $ defined as
\begin{equation}\label{eq:def_Uinv_(p,q)}
	U= \bigotimes_{i=1}^{L}u_{i}\,,
\end{equation}
the relations \eqref{eq:uinv_(p,q)_local_rel} imply that the Markov matrices of both processes are related by $ U $ as 
\begin{equation}\label{eq:Uinv_(p,q)_glob_rel}
	U M^{(\mathrm{msep})}= MU\,.
\end{equation}
With this last equation, we have
\begin{equation}\label{eq:rel_statio_(p,q)_(p+q)}
M \left( U\lvert S^{(\mathrm{msep})} \rangle\right)= U \left(  M^{(\mathrm{msep})}\lvert S^{(\mathrm{msep})} \rangle\right)=0\,.
\end{equation}
From the unicity of the stationary state of the $(p,q)$-SSEP with boundary $B^{(1)}$, we have that $ U \lvert S^{(\mathrm{msep})} \rangle $ is proportional to $ \lvert S \rangle $. We can easily verify that since the entries of $ \lvert S \rangle $ and $ \lvert S^{(\mathrm{msep})} \rangle $ sum to 1, we simply have
\begin{equation}\label{eq:rel_statio_(p,q)_(p+q)_2}
	\lvert S \rangle=U\lvert S^{(\mathrm{msep})} \rangle\,.
\end{equation}
Then, from the definition of $ u $ in Equation \eqref{eq:def_u_(p,q)}, we deduce that
\begin{equation}\label{eq:U_tau_action}
U \lvert \vec\tau \rangle=\bigotimes_{i=1}^{L}u_{i} \lvert \tau_{i} \rangle=\sum_{\vec\tau'\in \Phi^{-1}(\vec\tau)}\frac{1}{2^{n_{\mathrm{reac}}(\vec\tau')}} \lvert \vec\tau' \rangle\,.
\end{equation}
Finally, we deduce Equation \eqref{eq:statio_(p,q)_(p+q)} from the following 
\begin{align}\label{eq:US_(p+q)}
	U \lvert S^{(\mathrm{msep})} \rangle &=\sum_{\vec\tau\in \mathfrak{C}^{(\mathrm{msep})}} S^{(\mathrm{msep})}(\vec\tau)U \lvert \vec\tau \rangle\\
					   &=\sum_{\vec\tau\in \mathfrak{C}^{(\mathrm{msep})}}\sum_{\vec\tau'\in \Phi^{-1}(\vec\tau)}\frac{1}{2^{n_{\mathrm{reac}}(\vec\tau')}}S^{(\mathrm{msep})}(\vec\tau) \lvert \vec\tau' \rangle\\
&=\sum_{\vec\tau'\in \mathfrak{C}}\frac{1}{2^{n_{\mathrm{reac}}(\vec\tau')}}S^{(\mathrm{msep})}(\Phi(\vec\tau')) \lvert \vec\tau' \rangle\,.
\end{align}
\end{proof}

\begin{remark}
The projection of an exclusion process on another one with a fewer number of species is referred to in the literature as the identification or coloring procedure (see e.g. \cite{Ayyer2009}, \cite{CEMRV16}). The consistency of the coloring of reactive species in the $(p,q)$-SSEP with boundaries $B^{(1)}$ is formalised by the intertwining relations \eqref{eq:uinv_(p,q)_local_rel}.
\end{remark}

\begin{remark}\label{remark:statio_(p,q)_(p+q)}
Formula (\ref{eq:statio_(p,q)_(p+q)}) in the preceding proposition implies immediately the converse (but weaker) relation
\begin{equation}\label{eq:statio_(p+q)_(p,q)_converse}
\forall\vec\tau\in \mathfrak{C}^{(\mathrm{msep})}\,,\quad	S^{(\mathrm{msep})}(\vec\tau)=\sum_{\vec\tau'\in \Phi^{-1}(\vec\tau)}S(\vec\tau')\,.
\end{equation}
In fact, the result of the proposition is that all terms in the sum in (\ref{eq:statio_(p+q)_(p,q)_converse}) are  equal, hence the formula for any $S(\vec\tau)$. This is due to the fact that the Markov matrix $ M $ of the $ (p,q) $-SSEP with boundary $ B^{(1)} $ does not see any difference between species $a^{-}$ and $a^{+}$. This symmetry under the exchange $ a^{-}\leftrightarrow a^{+}$ (for any given $a\in\llbracket 1,q\rrbracket$) is reflected in Equation \eqref{eq:statio_(p,q)_(p+q)}.
\end{remark}

\begin{remark}
For a specific choice of the transition rates, the stationary distribution of the $ (p,q) $-SSEP with boundary $ B^{(1)} $ can similarly be related with the stationary distribution of the usual SSEP  (which has one particle species and empty sites).
\end{remark}

\subsubsection{Relations with physical quantities of the open multi-species SSEP}\label{ssec:phys_quant_q1}

In this part, we use the expression of the stationary state of the $ (p,q) $-SSEP with boundary $ B^{(1)} $ established previously to finish the proof of Proposition \ref{prop:physical_quant_q=q1}.

We first recall (see \cite{vanicat2017_2}) that in the multi-species SSEP with boundary $ B^{(\mathrm{msep})} $, the mean stationary currents and densities of particle species $ s\in\llbracket 1,p\rrbracket$ and $a\in \llbracket p+1,p+q\rrbracket$ are given by  
\begin{align}\label{eq:curr_dens_(msep)}
& \langle j_{s}\rangle_{i,i+1}^{(\mathrm{msep})}=\frac{\alpha_{s}-\delta_{s}}{L-1+\lambda+\mu},\qquad\langle s \rangle_{i}^{(\mathrm{msep})}=\frac{\alpha_{s}(L-i+\mu)+\delta_{s}(\lambda+i-1)}{L-1+\lambda+\mu}\,,\\
&\langle j_{a}\rangle_{i,i+1}^{(\mathrm{msep})}=\frac{2(\gamma_{a-p}-\beta_{a-p})}{L-1+\lambda+\mu},\qquad \langle a \rangle_{i}^{(\mathrm{msep})}=\frac{2\gamma_{a-p}(L-i+\mu)+2\beta_{a-p}(\lambda+i-1)}{L-1+\lambda+\mu}\,.
\end{align}
At this point, Proposition \ref{prop:physical_quant_q=q1} is implied by the following result.
\begin{proposition}
For $s\in\llbracket 1,p\rrbracket$, we have
\[\langle s\rangle_i=\langle s\rangle^{(\mathrm{msep})}_i\ \ \ \text{and}\ \ \ \langle j_s\rangle_{i,i+1}=\langle j_s\rangle^{(\mathrm{msep})}_{i,i+1}\ .\]
For $a\in\llbracket 1,q\rrbracket$, we have
\[\langle a^+\rangle_i=\langle a^-\rangle_i=\frac{1}{2}\langle p+a\rangle^{(\mathrm{msep})}_i\ \ \ \text{and}\ \ \ \langle j_{a^+}\rangle_{i,i+1}=\langle j_{a^-}\rangle_{i,i+1}=\frac{1}{2}\langle j_{p+a}\rangle^{(\mathrm{msep})}_{i,i+1}\ .\]
\end{proposition}

\begin{proof}
Firstly, using the symmetry between particle species $ a^{-} $ and $ a^{+} $, we have for $ s\in X $ and $ a\in \llbracket 1,q \rrbracket $ that
\begin{equation}\label{eq:ss'_a-a+_sym}
\forall i\in \llbracket 1,L-1 \rrbracket\,,\quad \langle a^{-}s \rangle_{i,i+1}=\langle a^{+}s \rangle_{i,i+1},\quad\langle sa^{-} \rangle_{i,i+1}=\langle sa^{+} \rangle_{i,i+1}\,.
\end{equation}
If $s$ is a singlet particle species, we have, using that the relevant transition rates are all equal to $1$,
\[\langle j_s\rangle_{i,i+1}=\sum_{\substack{s'\in X\\ s'\neq s}}\Bigl(\langle ss'\rangle_{i,i+1}-\langle s's\rangle_{i,i+1}\Bigr)=\sum_{s'\in X}\Bigl(\langle ss'\rangle_{i,i+1}-\langle s's\rangle_{i,i+1}\Bigr)=\langle s\rangle_{i}-\langle s\rangle_{i+1}\ .\]
Next, if $s=a^{\pm}$ for some $a\in\llbracket 1,q\rrbracket$, we have similarly
\[\langle j_{s}\rangle_{i,i+1}=\sum_{\substack{s'\in X\\ s'\neq a^{-},a^+}}\Bigl(\langle ss'\rangle_{i,i+1}-\langle s's\rangle_{i,i+1}\Bigr)=\sum_{s'\in X}\Bigl(\langle ss'\rangle_{i,i+1}-\langle s's\rangle_{i,i+1}\Bigr)=\langle s\rangle_{i}-\langle s\rangle_{i+1}\ .\]
For the first equality we used that $m(a^+a^-\rightarrow a^-a^+)=m(a^-a^+\rightarrow a^+a^-)=0$ (and all other relevant transition rates are $1$). For the second equality, we used (\ref{eq:ss'_a-a+_sym}) in the form $\langle a^+a^-\rangle_{i,i+1}=\langle a^-a^+\rangle_{i,i+1}$.

Thanks to the previous relations between currents and densities (which are also satisfied in the multi-species SSEP), it is enough to calculate the densities to prove the proposition.

If $s$ is a singlet particle species, we have in that case $\phi(s)=s$ and therefore,
\begin{align*}
	\langle s \rangle_{i}& =\sum_{\tau_{1},...,\tau_{i-1},\tau_{i+1},...,\tau_{L}\in X} S(\tau_{1},...,\tau_{i-1},s,\tau_{i+1},...,\tau_{L})\\
	&	=\sum_{\tau'_{1},...,\tau'_{i-1},\tau'_{i+1},...,\tau'_{L}\in X^{(\mathrm{msep})}} S^{(\mathrm{msep})}(\tau'_{1},...,\tau'_{i-1},\phi(s),\tau'_{i+1},...,\tau'_{L})=\langle s \rangle_{i}^{(\mathrm{msep})}\,,
\end{align*}
where we used in the second equality the relation
\begin{equation}\label{eq:statio_(p+q)_(p,q)_converse2}
\forall\vec\tau'\in \mathfrak{C}^{(\mathrm{msep})}\,,\quad	S^{(\mathrm{msep})}(\vec\tau')=\sum_{\vec\tau\in \Phi^{-1}(\vec\tau')}S(\vec\tau)
\end{equation}
between stationary states (see Remark \ref{remark:statio_(p,q)_(p+q)}). Using the same relation, we have for $a\in\llbracket 1,q\rrbracket$, 
\begin{align*}
	\langle a^{\pm} \rangle_{i}& =\sum_{\tau_{1},...,\tau_{i-1},\tau_{i+1},...,\tau_{L}\in X} S(\tau_{1},...,\tau_{i-1},a^{\pm},\tau_{i+1},...,\tau_{L})\\
	& =\frac{1}{2}\sum_{\tau_{1},...,\tau_{i-1},\tau_{i+1},...,\tau_{L}\in X} \Bigl(S(\tau_{1},...,\tau_{i-1},a^{\pm},\tau_{i+1},...,\tau_{L})+S(\tau_{1},...,\tau_{i-1},a^{\mp},\tau_{i+1},...,\tau_{L})\Bigr)\\
	&	=\frac{1}{2}\sum_{\tau'_{1},...,\tau'_{i-1},\tau'_{i+1},...,\tau'_{L}\in X^{(\mathrm{msep})}} S^{(\mathrm{msep})}(\tau'_{1},...,\tau'_{i-1},\phi(a^{\pm}),\tau'_{i+1},...,\tau'_{L})\\ & =\frac{1}{2}\langle p+a \rangle_{i}^{(\mathrm{msep})}\,,
\end{align*}
where in the second equality we used the symmetry of the stationary state under the exchange of $a^+$ and $a^-$, and in the remaining equalities, we used $\phi(a^{+})=\phi(a^{-})=p+a$.
\end{proof}

The fact that the mean reaction rate $\langle \nu_{a^\pm} \rangle_{i} $ of all reactive particle species vanishes for $ i\in \llbracket 2,L-1 \rrbracket $ is a direct consequence of its definition in Definition \ref{def:reac_rate} and the symmetry between particle species $ a^{-} $ and $ a^{+} $.

\subsection{Physical quantities in the case $q>q_{1}$}\label{ssec:case_q>q1}

In this last section, physical quantities in the case $q>q_1$ are discussed.

In the previous section \ref{ssec:case_q=q1}, the proof of Proposition \ref{prop:physical_quant_q=q1} relied on the exchange symmetry $ a^{+} \leftrightarrow a^{-} $ of two conjugated reactive particle species. More precisely, denoting by $ P_{\pm} $ the operator that exchanges $ \lvert a^{-} \rangle $ and $ \lvert a^{+} \rangle $ on each site of the lattice, the Markov matrix $ M $ of the $ (p,q) $-SSEP with boundaries $ B^{(1)} $ satisfies $ [P_{\pm},M]=0 $. In this case, the particle species $ a^{-} $ and $ a^{+} $ being indistinguishable, they can be identified and the process can be projected on the open multi-species SSEP.
However, for $ q>q_{1} $, the general form \eqref{Bblockdiagonal_Bbarblockdiagonal} of the boundaries $B$ and $\bar{B}$ leads in particular to transition rates of the form
\begin{equation}\label{eq:asym_rate}
	 a^- \xrightleftharpoons[\gamma_{a-q_1}^-/\lambda]{\gamma_{a-q_1}^+/\lambda}a^+,\quad a^- \xrightleftharpoons[\beta_{a-q_1}^-/\mu]{\beta_{a-q_1}^+/\mu}a^+,\ \ \ \ \ \forall a\in\llbracket q_1+1,q_1+q_2\rrbracket\,.
\end{equation}
If $ \gamma_{a-q_1}^{+}\neq\gamma_{a-q_1}^{-} $, the symmetry between species $ a^{-} $ and $ a^{+} $ is broken, they become distinguishable and $ [P_{\pm},M]\neq0 $. To compute all these physical quantities analytically, one needs other techniques, like for instance the Matrix Ansatz. We hope to come back on this point in a future work.

\begin{remark}
We have computed numerically for small values of $p,q$ and $L$ the mean currents, densities and reaction rates. We have checked in these examples that, contrary to the case $q=q_1$, reaction rates do not vanish in general for reactive particle species $b^{\pm}$ with $b>q_1$, and they depend in general on the lattice site. Additionally, the mean currents are also not constant along the lattice in general.
\end{remark}

\section{Outlook}
In this paper, we introduced a new family of integrable processes from set-theoretical solutions of the YBE that generalises the multi-species SSEP with periodic or open boundary conditions. 
The models are denoted $(p,q)$-SSEP where $p$ is the number of species behaving as in the multi-species SSEP, while $q$ encodes the number of pairs of reactive species, that behave differently. Several directions remain to be developed, and we now suggest some of them.

The first one is the determination of the stationary distribution in the open $(p,q)$-SSEP with a general integrable boundary $ B $. The Matrix Ansatz method (see e.g. \cite{Derrida1993}) could be applied to allow the determination of mean physical quantities. Note however that, due to the presence of sectors in the open case, the usual DEHP algebra \cite{DEHP1993} is believed to be not sufficient, and one needs to seek a wider algebra along the lines of what has been done for the DiSSEP model \cite{Crampe2014}.

The thermodynamic limit $L\to\infty$ and the connection with the Macroscopic Fluctuation Theory (MFT) are also desirable. In particular, for the second type of boundaries, since the model is dissipative, it would be interesting to see if the terms corresponding to creation and annihilation of particles  would appear in the MFT.

We have used integrability to find Boundary matrices that lead to simple expressions of the physical quantities \cite{Crampe2014}. Although we provide a large class of such integrable boundaries, a full classification of the solutions to the reflection equation remains to be done. This is desirable since integrable boundaries lead to tractable processes.

The next step in this direction would be to consider the open $(p,q)$-SSEP models associated to these boundaries. Again, a Matrix Ansatz would be desirable to compute the physical quantities of these models.

Considering the periodic $(p,q)$-SSEP, a combinatorial study of the sectors remains to be done, in the line of what we did for the $(p,1)$-SSEP in the appendix of the present paper. Then, it would be possible to look for possible quenches in the spirit of what has been done in \cite{ragoucy2026} for the twisted periodic multispecies SSEP.

Further generalizations for the bulk Markov matrix can also be pursued. From the physical point of view, adding free parameters to the set theoretical solution to the YBE is a natural question. This would imply a deformation of the set theoretical solutions, for instance in the spirit of \cite{Doikou_2021_2}, but interesting new physically relevant models are still lacking.
From the mathematical point of view, the $ (p,q) $-SSEP is a non-degenerate, involutive solution to the YBE that corresponds to a multipermutation solution of level 2 (see e.g. \cite{cedo2018}). It would be relevant to investigate exclusion processes constructed from  solutions that are more elaborate than this class of solutions (such as higher level multipermutation solutions or indecomposable solutions) and the corresponding brace structures.
Clearly, for any of these generalizations, the question of integrable boundaries should also be addressed.

\paragraph{Acknowledgments.} We acknowledge support from the CNRS IRP grant  \textsl{AAPT} and the USMB AAP grant  \textsl{MIMI}. LPdA warmly thanks the Laboratoire d'Annecy de Physique Théorique for  welcoming him during several short visits.

\appendix
\section{Stationary states and sectors for periodic models}\label{sec:sec_statio}

In this appendix, we show that the periodic $ (p,1) $-SSEP has several uniform stationary states.
Since they are uniform, we are reduced to the combinatorial study of the corresponding sectors defined below.

\subsection{Uniformity of stationary states}\label{ssec:uniform_statio}

We briefly recall the definition of stationary states and sectors. The unicity of stationary states per sector and the uniformity of their distribution will follow.

\begin{definition}[Stationary state]
A stationary state $\lvert S \rangle=\sum_{\vec \tau\in \mathfrak{C}} S(\vec \tau)\lvert \vec \tau\rangle $ is a probability vector (the coefficients are non-negative and sum to $1$) satisfying
\begin{equation}\label{eq:S_def}
 M\lvert S \rangle=0\ .
\end{equation}
\end{definition}

From Equation \eqref{eq:mast_eq}, a stationary state is independent of time. It corresponds to the state that the process reaches when $ t \rightarrow \infty $. 

The space of configurations of the process can be partitioned into sectors.

\begin{definition}[Sector]\label{defi_sector}
For a Markov process defined on the configuration space $ \mathfrak{C} $ with transitions given by the Markov matrix $ M $, the sector of a configuration $ \vec \tau\in \mathfrak{C} $ corresponds to the equivalence class of $ \vec \tau $ for the following equivalence relation :
\begin{equation}
	\vec \tau \sim \vec \tau' \Leftrightarrow \exists k,l\in \mathbb{N},\ \langle \vec\tau \rvert M^{k} \lvert \vec\tau' \rangle\neq 0\ \mathrm{and}\ \langle \vec\tau' \rvert M^{l} \lvert \vec\tau \rangle\neq 0\, .
\end{equation}
\end{definition}

In other words, there is a nonzero probability that the system can transition between any two configurations of a sector in a finite time. In the literature of Markov processes, it corresponds to a closed communicating class. Additionally, a Markov process with a single sector is called irreducible.

The unicity of the stationary state per sector is a well-known property of positive recurrent Markov processes (see e.g. \cite[th.3.5.3,~p.118]{Norris1998}), which holds for our current model because the space of configurations is finite.

When there are several sectors, we label them by  $\gamma$ and denote by $C_\gamma$ the set of all configurations in sector $\gamma$. We denote $\Gamma$ the set of all labels $\gamma$, and the set of all sectors is denoted by
\begin{equation}
    \mathcal{S}=\{C_{\gamma}\,,\ \gamma\in\Gamma\}\,.
\end{equation}
By expressing Equation \eqref{eq:mast_eq} on a component $ \vec \tau\in C_{\gamma} $ and with Equations \eqref{coeffdiagM}, \eqref{eq:sym_rate}, we can easily show that
\begin{equation}\label{S_gamma}
	\lvert S_{\gamma}\rangle=\frac{1}{\lvert C_{\gamma} \rvert}\sum_{\vec{\tau}\in C_\gamma} \lvert\vec{\tau}\rangle\ ,
\end{equation}
is the stationary state for the sector $ {\gamma} $.
In other words, if the process starts from a configuration in $ C_{\gamma} $, the probability that the process reaches any configuration of $ {\gamma} $ for asymptotically long times is equal to the inverse of the size of $ C_{\gamma} $. We deduce from the symmetry of transition rates \eqref{eq:sym_rate} that the system reaches thermodynamic equilibrium \cite{Mallick2015}.

\subsection{Sectors}\label{ssec:sectors}

To complete the study of the long-time behaviour of the periodic $ (p,1) $-SSEP, it remains to determine the sectors ${\gamma} $ and their cardinality $|C_\gamma|$. We will uniquely characterise a sector by a set of invariants for the dynamics and give for each sector a natural representative.

\subsubsection{Invariants of a configuration}

In the following, for a given configuration $ \vec \tau $, we denote by $ n_{0}(\vec\tau)$, $ n_{\bar{0}}(\vec\tau) $, $ n_{1}(\vec\tau) $, ...,  $ n_{p}(\vec\tau) $ the number of particles of species $ 0 $, $ \bar{0} $, $1$, ..., $p$ respectively, namely, for $a\in\{0,\bar0,1,\dots,p\}$,
\[n_a(\vec\tau)=|\{i\in\llbracket 1,L\rrbracket\ \text{such that}\ \tau_i=a\}|\ .\]
We also define
\[n_{0\bar 0}(\vec\tau)=n_0(\vec\tau)+n_{\bar 0}(\vec\tau)\ .\]
Of course we have $L=n_{0\bar 0}(\vec\tau)+n_{1}(\vec\tau)+\dots+n_{p}(\vec\tau)$.
\begin{definition}[profile and parity]\label{def:profpar}
The profile $p(\vec\tau)$ of a configuration $\vec\tau$ is defined by
\[p(\vec\tau)=(n_{0\bar 0}(\vec\tau),n_{1}(\vec\tau),\dots,n_{p}(\vec\tau))\ .\]
The parity $\epsilon(\vec\tau)$ of a configuration $\vec\tau$ is defined by
\[\epsilon(\vec\tau)=\left\{\begin{array}{ll} -1 & \text{if $n_{\bar 0}(\vec\tau)$ is odd,}\\
+1 & \text{if $n_{\bar 0}(\vec\tau)$ is even.}\end{array}\right.\]
\end{definition}

Naturally, the profile and the parity are going to be invariants under the dynamics. However, in general, more invariants will be needed to characterize a sector and the additional invariants are related to the positions of the particles $0$ and $\bar 0$. Therefore in the following it will be useful to consider the subconfigurations where all particles of species $1,\dots,p$ are removed. This leads to the following definition.

\begin{definition}[$0\bar 0$-subconfiguration]\label{def:subconf}
For a configuration $ \vec \tau=(\tau_{1},\tau_{2},\dots,\tau_{L}) $, the corresponding $0\bar 0$-subconfiguration $ \vec \tau'=(\tau_{1}',\tau_{2}',\dots,\tau_{n_{0\bar 0}(\vec\tau)}') $ of length $ n_{0\bar 0}(\vec\tau)$ is the configuration obtained from $ \vec \tau $ by removing all particles of species $1,\dots,p$ while preserving the ordering of the particles of species 0 and $ \bar{0} $.
\end{definition}
We are ready to define the last invariant. Note that we define it only when $ n_{0\bar 0}(\vec\tau)$ is even because this is the only situation where it will be needed.
\begin{definition}[Difference]\label{def:diff}
Let $ \vec \tau$ be a configuration with $ n_{0\bar 0}(\vec\tau)$ even and denote the corresponding $0\bar 0$-subconfiguration by $\vec \tau'=(\tau_{1}',\tau_{2}',\dots,\tau_{n_{0\bar 0}(\vec\tau)}')$. We let
\[n_{\bar 0}^{(e)}(\vec\tau)=\ \text{number of $\bar 0$ in $(\tau'_2,\tau'_4,\dots,\tau'_{n_{0\bar0}(\vec\tau)})$},\]
\[n_{\bar 0}^{(o)}(\vec\tau)=\ \text{number of $\bar 0$ in $(\tau'_1,\tau'_3,\dots,\tau'_{n_{0\bar0}(\vec\tau)-1})$},\]
and we define the difference $\Delta_{\bar 0}(\vec\tau)$ by
\[\Delta_{\bar 0}(\vec\tau)=n_{\bar 0}^{(e)}(\vec\tau)-n_{\bar 0}^{(o)}(\vec\tau)\ .\]
\end{definition}
\begin{example}
Let $\vec\tau=(1,\bar 0,0,2,\bar 0,1,\bar 0,0,0,2)$. We have $\vec\tau'=(\bar 0,0,\bar 0,\bar 0,0,0)$ and $n_{\bar 0}^{(e)}(\vec\tau)=1$ and $n_{\bar 0}^{(o)}(\vec\tau)=2$. Therefore $\Delta_{\bar 0}(\vec\tau)=-1$.
\end{example}

\subsubsection{Characterisation of the sectors}

\begin{proposition}
Let $\vec\tau$ be a configuration. The sector to which $\vec\tau$ belongs is uniquely characterised by the profile $\vec p(\vec\tau)=(n_{0\bar 0}(\vec\tau),n_{1}(\vec\tau),\dots,n_{p}(\vec\tau))$ together with:
\begin{itemize}
    \item[(i)] the parity $\epsilon(\vec\tau)$ if $n_{0\bar0}(\vec\tau)$ is odd;
    \item[(ii)] the absolute value $|\Delta_{\bar0}(\vec\tau)|$ of the difference if $n_{0\bar0}(\vec\tau)$ is even and $n_{0\bar0}(\vec\tau)<L$;
    \item[(iii)] the difference $\Delta_{\bar0}(\vec\tau)$ if $n_{0\bar0}(\vec\tau)=L$ and $L$ is even.
\end{itemize}
\end{proposition}
In other words, two configurations belong to the same sector if and only if they have the same invariants as described in the proposition. We will provide canonical representatives of each sector during the proof below, namely in (\ref{rep:ia})-(\ref{rep:ib}) in Case (i), in (\ref{rep:ii}) in Case (ii) and in (\ref{rep:iiia})-(\ref{rep:iiib}) in Case (iii).
\begin{proof}
We start by showing that the claimed invariants are indeed invariants. First the profile $p(\vec\tau)$ and the parity $\epsilon(\vec\tau)$ are invariants since the only moves that change the number of particles of a given species are $00\to \bar0\bar0$ and $\bar0\bar0\to 00$, and they obviously leave invariants the total number $n_{0\bar 0}(\vec\tau)$ of particles of species $0$ and $\bar 0$ and the parity of the number of $\bar 0$ particles.

Now assume that $n_{0\bar0}(\vec\tau)=L$ and that $L$ is even. Since there are no particle of species $1,\dots,p$ in $\vec\tau$, the only non-trivial moves are $00\to \bar0\bar0$ and $\bar0\bar0\to 00$. Since $L$ is even, whether such a move happens between sites $i,i+1$ or between sites $L,1$, this always shifts $n_{\bar 0}^{(e)}(\vec\tau)$ and $n_{\bar 0}^{(o)}(\vec\tau)$ by the same amount ($+1$ or $-1$) and therefore in this case the difference $\Delta_{\bar0}(\vec\tau)$ is invariant.

Finally, if $n_{0\bar0}(\vec\tau)$ is even and $n_{0\bar0}(\vec\tau)<L$, the same reasoning shows that the moves $00\to \bar0\bar0$ and $\bar0\bar0\to 00$ leave the difference $\Delta_{\bar0}(\vec\tau)$ invariant. However, we also need to consider moves of the form $0a\leftrightarrow a0$ and $\bar 0a\leftrightarrow a\bar 0$ for $a\in\llbracket 1,p\rrbracket$. If such a move happens between sites $i,i+1$, it leaves the $0\bar0$-subconfiguration invariant, hence also the difference $\Delta_{\bar0}(\vec\tau)$. However, if it happens between sites $L$ and $1$, it has the effect of shifting all entries (cyclically) of the $0\bar0$-subconfiguration by $1$. Therefore, since $n_{0\bar 0}(\vec\tau)$ is even, $n_{\bar 0}^{(e)}(\vec\tau)$ becomes equal to $n_{\bar 0}^{(o)}(\vec\tau)$ and vice versa. Thus, in this case, only the absolute value $|\Delta_{\bar0}(\vec\tau)|$ is invariant.

Having proved that the claimed invariants are indeed invariants, it remains to show that a sector is uniquely characterised by them.

Let $\vec\tau$ be a configuration. First using moves of the form $a0\to 0a$ and $a\bar 0\to \bar 0a$, we can move all particles of species $a\in\llbracket 1,p\rrbracket$ to the right. Then using moves $\bar 0\bar0\to 00$ and $00\bar0\to \bar0\bar0\bar0\to \bar000$, we are left with no two adjacent $\bar0$'s and no $\bar0$ preceded by two consecutive $0$'s. Now we deal with the different cases of the proposition.

(iii) Assume that $n_{0\bar0}(\vec\tau)=L$ and $L$ is even so that the profile is $p(\vec\tau)=(L,0,\dots,0)$. At this point we are left with one of the following configurations
\begin{equation}\label{rep:iiia}
\underbrace{0\bar00\bar0\dots0\bar0}_{2b}00\dots 0\ \ \ \ \ :\ \Delta_{\bar 0}(\vec\tau)=b\ ,
\end{equation}
\begin{equation}\label{rep:iiib}
\underbrace{\bar00\bar00\dots\bar00}_{2b}00\dots 0\ \ \ \ \ :\ \Delta_{\bar 0}(\vec\tau)=-b\ ,
\end{equation}
where $b\in\llbracket 0,L/2\rrbracket$. These configurations correspond to different values of the invariant $\Delta_{\bar 0}(\vec\tau)$ and therefore to different sectors. We have shown that any configuration $\vec\tau$ in Case (iii) is equivalent to a canonical representative in (\ref{rep:iiia})-(\ref{rep:iiib}) uniquely determined by the invariants. This proves the proposition in Case (iii).

(ii) Assume that $n_{0\bar0}(\vec\tau)$ is even and that $n_{0\bar0}(\vec\tau)<L$. At this point, we are left with configurations as in (\ref{rep:iiia})-(\ref{rep:iiib}) with particles of species $1,\dots,p$ added to the right. Assume that we have the configuration
\[\underbrace{\bar00\bar00\dots\bar00}_{2b}\underbrace{0\dots 0}_{a}\underbrace{1\dots1}_{n_1(\vec\tau)}\dots\dots \underbrace{p\dots p}_{n_p(\vec\tau)}\ \ \ \ \text{with $2b+a=n_{0\bar0}(\vec\tau)$.}\]
Using periodicity and moves $\bar0 a\to a \bar 0$, we can move the first $\bar0$ on the left across all the particles of species $1,\dots,p$ to put it at the end of the sequence of $0$'s. The resulting configuration starts with some $i\in\llbracket 1,p\rrbracket$ that we can move back to the right to its original position. So the result is just as if we had taken the first $\bar 0$ and put it just after the sequence of $0$'s. Now since $a$ is even, we can use several times $00\bar0\to \bar0\bar0\bar0\to \bar000$ to move the $\bar 0$ across this sequence, and we have obtained a configuration:
\begin{equation}\label{rep:ii}\underbrace{0\bar00\bar0\dots0\bar0}_{2b}\underbrace{0\dots 0}_{a}\underbrace{1\dots1}_{n_1(\vec\tau)}\dots\dots \underbrace{p\dots p}_{n_p(\vec\tau)}\ \ \ \ \ :\ |\Delta_{\bar 0}(\vec\tau)|=b\,,
\end{equation}
where $b\in\llbracket 0,n_{0\bar0}(\vec\tau)/2\rrbracket$. These configurations correspond to different values of the invariant $|\Delta_{\bar 0}(\vec\tau)|$ and therefore to different sectors. We have shown that any configuration $\vec\tau$ in Case (ii) is equivalent to a canonical representative in (\ref{rep:ii}) uniquely determined by the invariants.

(i) Finally assume that $n_{0\bar0}(\vec\tau)$ is odd. At this point, we have a configuration with $0\bar0$-subconfiguration of the form
\[\underbrace{\bar00\bar00\dots\bar00}_{2b}\underbrace{0\dots 0}_{a}\ \ \ \ \ \text{or}\ \ \ \ \ \underbrace{0\bar00\bar0\dots0\bar0}_{2b}\underbrace{0\dots 0}_{a}\,.\]
Note that in what follows we can ignore the possible presence of particles of species $1,\dots,p$ since we can always move them arbitrarily across particles of species $0$ and $\bar0$. Now, assume that there are at least two particles of species $\bar0$. Since $n_{0\bar0}(\vec\tau)=2b+a$ is odd, we have that $a$ is odd and then $a+1$ is even. Therefore, as in case (ii), we can use several times the move $00\bar0\to \bar0\bar0\bar0\to \bar000$ to move the first $\bar0$ across the sequence of $0$'s and put it next to the last $\bar0$. Then with a move $\bar0\bar0\to00$, we can reduce the numbers of $\bar0$ by $2$. Therefore we are left with at most $1$ particle of species $\bar 0$. Moreover, if  there is one $\bar0$, using the same move as before we can shift its position by $2$ in any direction in the $0\bar0$-subconfiguration. Since by assumption $n_{0\bar0}(\vec\tau)$ is odd, it means that we can always put the only $\bar0$ in the first position. In conclusion, we are left with the only two following possibilities
\begin{equation}\label{rep:ia}\underbrace{0\dots0}_{n_{0\bar0}(\vec\tau)}\underbrace{1\dots1}_{n_1(\vec\tau)}\dots\dots \underbrace{p\dots p}_{n_p(\vec\tau)}\ \ \ \ \ :\ \epsilon(\vec\tau)=+1\,,
\end{equation}
\begin{equation}\label{rep:ib}\underbrace{\bar00\dots0}_{n_{0\bar0}(\vec\tau)}\underbrace{1\dots1}_{n_1(\vec\tau)}\dots\dots \underbrace{p\dots p}_{n_p(\vec\tau)}\ \ \ \ \ :\ \epsilon(\vec\tau)=-1\,,
\end{equation}
These configurations correspond to different values of the invariant $\epsilon(\vec\tau)$ and therefore to different sectors. We have shown that any configuration $\vec\tau$ in Case (i) is equivalent to a canonical representative in (\ref{rep:ia})-(\ref{rep:ib}) uniquely determined by the invariants.
\end{proof}

\begin{remark}
Let $p=1$. A straightforward calculation using the previous characterisation gives that the total number of sectors $ \lvert \mathcal{S} \rvert $ is given by
\begin{equation}\label{eq:sec_nb}
		\lvert \mathcal{S} \rvert=
		\begin{dcases}
		    2L+1+\frac{L(L+2)}{8}\quad \mathrm{for}\ L\ \mathrm{even}\,,\\
		    L+1+\frac{(L+3)(L+1)}{8}\quad \mathrm{for}\ L\ \mathrm{odd}\,.
		\end{dcases}
	\end{equation}
In \cite[Remark 6.3]{ragoucy2026}, we claimed that the twisted SSEP with $ N=3 $ species is not equivalent to the $ (1,1) $-SSEP for all $ L>2 $. This can be checked by comparing the number of sectors given by Formula \eqref{eq:sec_nb} and Formula (4.7) for $ N=3 $ in \cite{ragoucy2026}.
\end{remark}

\subsubsection{Cardinalities of sectors}

Having characterised the different sectors, it remains to calculate their cardinalities. Recall that for a profile $\vec p(\vec\tau)=(n_{0\bar0}(\vec\tau),n_1(\vec\tau),\dots,n_p(\vec\tau))$, we have $n_{0\bar0}(\vec\tau)+n_1(\vec\tau)+\dots+n_p(\vec\tau)=L$. We define the corresponding multinomial coefficients:
\[\binom{L}{\vec p(\vec\tau)}=\frac{L!}{n_{0\bar0}(\vec\tau)!n_1(\vec\tau)!\dots n_p(\vec\tau)!}\ .\]
\begin{proposition}
Let $\vec p=(n_{0\bar0},n_1,\dots,n_p)\in(\mathbb{Z}_{\geq 0})^{p+1}$ with $n_{0\bar0}+n_1+\dots+n_p=L$ be a profile.
\begin{itemize}
    \item [(i)] If $n_{0\bar0}$ is odd, the cardinality of the sector indexed by the profile $\vec p$ and the parity $\epsilon\in\{\pm1\}$, is:
\[|C_{\vec p,\epsilon}|=\binom{L}{\vec p}2^{n_{0\bar0}-1}\ .\]
\item [(ii)] If $n_{0\bar0}$ is even and $n_{0\bar0}<L$, the cardinality of the sector indexed by the profile $\vec p$ and by the absolute value of the difference $|\Delta_{\bar0}|\in \llbracket 0,n_{0\bar0}/2\rrbracket$ is:
\[|C_{\vec p,|\Delta_{\bar0}|}|=\left\{\begin{array}{ll}
\displaystyle\binom{L}{\vec p}\binom{n_{0\bar0}}{n_{0\bar0}/2} & \text{if $|\Delta_{\bar0}|=0$,}\\[1em]
\displaystyle2\binom{L}{\vec p}\binom{n_{0\bar0}}{n_{0\bar0}/2+|\Delta_{\bar0}|}& \text{if $|\Delta_{\bar0}|>0$.} 
 \end{array}\right.\]
\item [(iii)] If $n_{0\bar0}=L$ and $L$ is even, the cardinality of the sector indexed by the profile $\vec p=(L,0,\dots,0)$ and by the difference $\Delta_{\bar0}\in \llbracket -L/2,L/2\rrbracket$ is: \[|C_{\vec p,\Delta_{\bar0}}|=\binom{L}{L/2+|\Delta_{\bar0}|}\ .\]
\end{itemize}
\end{proposition}
\begin{proof}
Given a profile $\vec p$, there are $\binom{L}{\vec p}$ ways of distributing the particles of species $1,\dots,p$ across the $L$ sites. There only remains to count the number of ways of distributing the particles $0$ and $\bar 0$ in the remaining sites, taking into account the other invariant.

(i) Given a parity $\epsilon\in\{\pm1\}$, we can choose arbitrarily $0$ or $\bar 0$ on all the remaining sites but one (the choice for the last one being dictated by the parity). Therefore there are $n_{0\bar0}-1$ binary choices to make.

(ii)-(iii) Let $\vec\tau$ be a configuration with an even $n_{0\bar 0}(\vec\tau)$ and with a given value of $\Delta_{\bar 0}(\vec\tau)$. The set $\{n_{\bar 0}^{(e)}(\vec\tau),n_{\bar 0}^{(o)}(\vec\tau)\}$ (recall Definition \ref{def:diff}) is of the form $\{i,i+|\Delta_{\bar 0}(\vec\tau)|\}$ for some $i$. Then the number of ways of distributing the particles of species $0$ and $\bar 0$ corresponding to this choice of $i$ is:
\[\binom{n_{0\bar 0}(\vec\tau)/2}{i}\binom{n_{0\bar 0}(\vec\tau)/2}{i+|\Delta_{\bar 0}(\vec\tau)|}\ .\]
Therefore, for a given value of $\Delta_{\bar 0}(\vec\tau)$, we must sum the previous numbers for all allowed values of $i$. It is well-known (Vandermonde identity) that this gives $$\binom{n_{0\bar 0}(\vec\tau)}{n_{0\bar 0}(\vec\tau)/2+|\Delta_{\bar 0}(\vec\tau)|}\ .$$
In particular, the number of configurations with a given value of $\Delta_{\bar 0}(\vec\tau)$ only depends on $|\Delta_{\bar 0}(\vec\tau)|$. The above calculation accounts for the formulas in Cases (ii) and (iii). Note that in Case (ii) when $|\Delta_{\bar0}|>0$, the factor $2$ comes from the fact that two opposite values for $\Delta_{\bar0}$ are possible.
\end{proof}

\end{document}